\documentclass[11pt]{article}
\usepackage[bottom=1in, left=1in, right=1in, top=1in]{geometry}

\pdfoutput=1 

\usepackage[utf8]{inputenc}
\usepackage[T1]{fontenc}
\usepackage{lmodern}
\usepackage{microtype}

\usepackage[USenglish]{babel}

\usepackage[ruled, vlined, linesnumbered, nofillcomment]{algorithm2e}

\DontPrintSemicolon
\SetKwInOut{Input}{Input}
\SetKwInOut{Output}{Output}
\SetKwProg{Procedure}{Procedure}{}{}

\usepackage{graphicx} 
\usepackage{subcaption}
\usepackage{xcolor}
\usepackage{amsmath,amssymb,amsthm}
\usepackage{enumerate}
\usepackage{mathtools}
\usepackage[]{hyperref}
\hypersetup{
    colorlinks,
    linkcolor={red!50!black},
    citecolor={blue!50!black},
    urlcolor=.
}
\usepackage{float}
\usepackage[capitalise]{cleveref}
\usepackage{todonotes}
\usepackage{tabularx}
\usepackage[square,numbers]{natbib}
\usepackage{tikz}
\usetikzlibrary{calc}
\usetikzlibrary{positioning,backgrounds,patterns,calc,matrix,shapes,decorations.pathreplacing,decorations.pathmorphing,math}
\tikzset{crossing/.style={cross out, draw=red, minimum size=2*(#1-\pgflinewidth), inner sep=0pt, outer sep=1pt, very thick}, crossing/.default={4pt}}
\usetikzlibrary{arrows.meta,bending,decorations.markings,hobby,patterns,calc}
\newcolumntype{C}[1]{>{\centering\let\newline\\\arraybackslash\hspace{0pt}}m{#1}}

\graphicspath{ {./figs/} }

\usepackage{thmtools, thm-restate}

\newtheorem{claim}{Claim}
\newtheorem{lemma}{Lemma}

\newtheorem{observation}{Observation}

\newtheorem{theorem}{Theorem}

\newcommand{\problemdef}[3]{
	\begin{center}
	\begin{minipage}{0.95\columnwidth}
		\noindent
		\textsc{#1}
		\vspace{5pt}\\
		\setlength{\tabcolsep}{3pt}
		\begin{tabularx}{\textwidth}{@{}lX@{}}
			\textbf{Input:}     & #2 \\
			\textbf{Question:}  & #3
		\end{tabularx}
	\end{minipage}
	\end{center}
}
\newcommand{\optproblemdef}[3]{
	\begin{center}
	\begin{minipage}{0.95\columnwidth}
		\noindent
		#1
		\vspace{5pt}\\
		\setlength{\tabcolsep}{3pt}
		\begin{tabularx}{\textwidth}{@{}lX@{}}
			\textbf{Input:}     & #2 \\
			\textbf{Task:}  & #3
		\end{tabularx}
	\end{minipage}
	\end{center}
}

\newcommand{\problem}[1]{\textnormal{\textsc{#1}}}
\newcommand{\CRI}{\problem{Closeness Ratio Improvement}}
\newcommand{\CGM}{\problem{Closeness Gap Minimization}}
\newcommand{\SC}{\problem{Set Cover}}

\newcommand{\APCRI}{\problem{All-Pairs Closeness Ratio Improvement}}
\newcommand{\allpairs}{\APCRI}

\newcommand{\ccg}[2]{\textnormal{cc}_{{#1}}({#2})}
\newcommand{\ccrg}[3]{\textnormal{cc-ratio}_{{#1}}({#2}, {#3})}
\newcommand{\distg}[3]{\text{d}_{{#1}}({#2}, {#3})}
\newcommand{\meandistg}[3]{\tilde{d}_{{#1}}({#2}, {#3})}

\newcommand{\cclass}[1]{\textnormal{\textsf{#1}}}
\newcommand{\opt}{\textnormal{\textsf{opt}}}

\newcommand{\bestcc}[2]{\textsc{best-cc}_{#1}(#2)}
\newcommand{\worstcc}[2]{\textsc{worst-cc}_{#1}(#2)}
\newcommand{\domset}[0]{\textsc{Dominating Set}}
\newcommand{\dsum}[1]{d^\textup{sum}_{#1}}
\newcommand{\edgetype}[2]{#1 \times #2}

\newcommand{\ccworst}[2]{\textnormal{cc-worst}_{{#1}}({#2})}

\newcommand{\constructGraph}{H}

\title{Equalizing Closeness Centralities via Edge Additions}
\usepackage{authblk}
\author[1]{Alex Crane}
\author[2]{Sorelle A. Friedler}
\author[2]{Mihir Patel}
\author[1]{Freda Shi}
\author[1]{\authorcr Blair D. Sullivan}
\author[1]{Madelyn J. Warr}

\affil[1]{University of Utah}
\affil[2]{Haverford College}
\date{}

\begin{document}

\maketitle

\begin{abstract}
    Graph modification problems with the goal of optimizing some measure of a given node's network position have a rich history in the algorithms literature. Less commonly explored are modification problems with the goal of \emph{equalizing} positions, though this class of problems is well-motivated from the perspective of equalizing social capital, i.e., algorithmic fairness.
    In this work, we study how to add edges to make the closeness centralities of a given pair of nodes more equal.
    We formalize several versions of this problem: \CRI, which aims to maximize the ratio of closeness centralities between two specified nodes, and \CGM, which aims to minimize the absolute difference of centralities. 
    For the former, we present a quasilinear-time $\frac{6}{11}$-approximation, complemented by a bicriteria inapproximability bound.
    In contrast to this positive result, we show that \CGM{} admits no multiplicative approximation, unless $\cclass{P} = \cclass{NP}$.
    We also establish \cclass{NP}-hardness for~\APCRI{}, which aims to maximize the minimum ratio of closeness centralities across all node pairs.
\end{abstract}

\section{Introduction}

An important class of problems in the \emph{algorithmic fairness} literature presents as input a network (a graph, possibly with some algorithmically relevant metadata) and asks for an efficiently computed set of~\emph{network interventions} meant either to alter the social dynamics of the network, or to reveal already present hidden structure. Generally, the motivating assumption is that real-world resources accrue to people in prominent positions within social networks, and that the distribution of this so-called \emph{social capital} hinges on network structure itself~\cite{granovetter1973strength}. For a broad survey, see~\cite{saxena2024fairsna}.
Work in this direction has generally proceeded along several (occasionally overlapping) lines of inquiry.

The first has focused on studying classic combinatorial optimization problems with \emph{fairness constraints}, and exploring the algorithmic cost of these constraints by either proving impossibility results or lifting algorithmic techniques used for the classic unconstrained version of the problem.
Fair clustering has received special emphasis, e.g.,~\cite{ahmadi2020fair,ahmadian2020fair,ahmadian2023improved,abbasi2021fair,ghadiri2021socially,vakilian2022improved,davies2024simultaneous}, especially since the seminal work of Chierichetti~\emph{et al.}~\cite{chierichetti2017fair}.
Often, relatively tractable combinatorial problems are significantly harder when the solution space is constrained to enforce fairness, e.g.,~\cite{casel2023fair},
though there have been some algorithmic success stories involving bicriteria approximations, e.g.~\cite{bercea2019cost,abbasi2021fair,crane2025edge}, or dual parameterizations, e.g.,~\cite{komusiewicz2012cluster}. 
Somewhat more successful has been the trend of encoding fairness via a modified objective function, rather than by imposing additional constraints
on a classic objective, e.g.,~\cite{froese2024modification,abbasi2021fair}. With this perspective, a series of works on minimax or $\ell_p$-norm correlation clustering (beginning with~\cite{puleo2016correlation}, most recently~\cite{davies2024simultaneous}) as well as modification with ``degree constraints,'' e.g.,~\cite{komusiewicz2012cluster,adriaens2022diameter}, can be seen through a fairness lens, even if fairness was not the original motivation.
This line of work has been primarily theoretical, although an important practical insight, seen in e.g.,~\cite{bercea2019cost,davies2024simultaneous,crane2025edge}, is that the \emph{cost} of fairness is often low in the sense that fair solutions are also often high-quality when evaluated by traditional criteria.

Another line of work is distinguished by an empirical focus 
on information access, combined with some theoretical analysis, often in the context of some non-trivial diffusion model such as Independent Cascade~\cite{kempe03maximizing}. This has given rise to fair variants of, e.g., the classic influence maximization problem~\cite{fish,stoica2019fairness}. 
Interventions in this setting could include selecting vertices at which to seed information (in which case the intervention is a modification of the metadata, rather than the underlying graph) or modifying the graph directly through e.g., edge addition, including formal or empirical analysis of the information spreading properties of proposed heuristics~\cite{bashardoust, stoica2020seeding, fish, stoica2019fairness, liu2021rawlsnet, jalali2023fairness}. 
While the majority of works in this literature have introduced and experimentally evaluated heuristics, recently a greater interest in the analysis of algorithms has emerged, e.g.,~\cite{bhaskara2024optimizing, rui2023scalable}.
This network/graph focused work builds upon a broader literature studying fairness in machine learning. This literature has considered fairness both through added constraints (e.g., \cite{hardt2016equality}) and via modified objective functions (e.g., \cite{friedler2019comparative}) where 
some emphasis has been given to understanding which classes of objective functions are most socially desirable from a fairness perspective, e.g.,~\cite{yeh2024analyzing}.

Both lines of work can also be divided into approaches which emphasize \emph{individual fairness} by optimizing a maximin or minimax objective over all nodes, e.g.,~\cite{fish,puleo2016correlation}, or by imposing a per-node fairness constraint, e.g.,~\cite{komusiewicz2012cluster}, and approaches which emphasize \emph{group fairness} by optimizing various measures based on demographic node attributes, e.g.,~\cite{farnad2020unifying,stoica2019fairness}.\looseness=-1

In this work, we take a group fairness approach and directly optimize a fairness objective, but we do so in a purely graph-theoretic context, and our approach is closely aligned with a classic discrete algorithms perspective.
Instead of considering a complex model of information flow, we will use \emph{centrality} to measure a node's social capital, e.g., in the spirit of~\cite{borgatti1998network}.
Though many notions of centrality exist, we focus on \emph{closeness centrality},
which is simple to understand, exhibits many qualities desirable in a centrality measure\footnote{In our context (undirected connected graphs), the ``main problem'' with closeness centrality identified by~\cite{boldi2014axioms} (non-reachability) does not occur.}~\cite{boldi2014axioms},
and is very popular in the recent algorithms literature, e.g.,~\cite{bergamini2016,hoffmann2017parameterized,bergamini2018scaling,staus2023exact}.
Thus, closeness is a natural starting point for our goal of algorithmically equalizing node centralities.

Formally, we consider a simple, unweighted, undirected graph $G$ with two specified vertices $a, b$ and an integer $k$.
Our task is to add at most $k$ edges to $G$ while maximizing the ratio of the centralities of $a$ and $b$, i.e., the best possible objective value is $1$. Here,
the closeness centrality $\ccg{G}{v}$ of a vertex $v$ is the sum of $v$'s shortest-path distances to all vertices.\footnote{Closeness is also commonly defined as the reciprocal of this quantity. In the context of our objective (a ratio of closeness centralities), there is no meaningful distinction between these definitions. The impossibility result we establish for a gap-variant (\Cref{cor:gap-minimization-hard}) also holds for both definitions.}
We note that the restriction to just two specified vertices of interest is a severe one, but also not dissimilar to many algorithmic works which impose such restrictions to make steps toward a more general goal, e.g., group fairness works restricted to two groups~\cite{froese2024modification, ahmadi2020fair}, or partitioning problems restricted to two clusters~\cite{giotis2006correlation}.

We formalize two extensions of our two-vertex setting in~\Cref{sec:preliminaries}. The first directly generalizes the setting described above by asking that we maximize a ratio of closeness centralities between two disjoint vertex subsets. We give some commentary in~\Cref{sec:conclusion} regarding the difficulties in extending our algorithmic techniques to this setting.
The second extension asks that we optimize the worst closeness centrality ratio across all vertex pairs, rather than a given pair. We establish a lower bound for this problem in~\Cref{sec:hardness-section}.

Naturally, our work is also closely related to the vast algorithmic literature on edge-modification in graphs; see e.g., a recent survey~\cite{crespelle2023survey}.
Especially relevant is the body of work on adding edges to optimize various notions of a given node's centrality~\cite{bergamini2016,crescenzi2016,shan2018,avrachenkov2006effect,olsen2014approximability}, or even to optimize a group's centrality~\cite{medya2018,staus2023exact}.
Conceptually, a major distinguishing feature of our work is that the ratio objective
appears not to be amenable (see~\Cref{appendix:strategies}) to approaches appearing frequently in the edge-modification literature, e.g., submodularity analysis or reductions to clustering problems~\cite{adriaens2022diameter,bilo,demaine,li,meyerson,frati2015augmenting}.
Thus, among our technical contributions is a demonstration of how to reason about the structure of hard instances
in a way that can be harnessed algorithmically, even when the objective function is a ratio.

\bigskip
\textbf{Our Results:} We introduce and study optimization problems related to balancing the closeness centrality of a given pair of nodes. In \CRI{}, the goal is to maximize the ratio of the closeness centralities of two given vertices with a budgeted number of edge additions. Our main result is a quasilinear-time $\frac{6}{11}$-approximation for this problem. We also prove that the problem is \cclass{NP}-hard for any target ratio $\tau$ in the interval $(\frac{1}{2}, 1]$, and trivial when $\tau \leq \frac{1}{2}$. Additionally, we show that it is \cclass{NP}-hard to approximate \CRI{} within any factor better than $\frac{5e}{5e + 1 - \varepsilon} \approx 0.932$, and we generalize this lower bound to give a $\frac{5e^c}{5e^c + 1 - \varepsilon}$ approximation bound for any algorithm which is allowed to add $ck$ edges, where $c \geq 1$ is a constant. In \CGM{}, the goal is to minimize the absolute difference between the closeness centralities of two nodes, again with a budgeted number of edge additions. We prove that this problem is \cclass{NP}-hard and, in contrast to \CRI{}, does not admit any multiplicative approximation algorithm. 
Finally, in~\APCRI{} the goal is to use a budgeted number of edge additions to maximize the minimum ratio of closeness centralities across all vertex pairs. We establish that this problem is~\cclass{NP}-hard.

\section{Preliminaries}\label{sec:preliminaries}
We introduce the problem of \CRI{} on a social network where people are modeled as vertices and edges represent connections between people. Specifically, let $G = (V, E)$ be an undirected graph with vertices $V$ and edges $E$. Our goal will be to add a budget of $k$ edges to $G$ to equalize the closeness centrality of two vertices. We denote a graph $G=(V,E)$ augmented with an edge set $S\subseteq V^2\backslash E$ as $G+S=(V,E\cup S)$. Within $G$, the shortest path length between two vertices $u,v\in V$ is represented as $\distg{G}{u}{v}$. Finally, we define the closeness centrality of a vertex $v\in V$ in a graph $G=(V,E)$ as $\ccg{G}{v}=\sum_{u\in V}\distg{G}{u}{v}$, the sum of $v$'s shortest-path distances to each vertex in the graph. Note that having a smaller closeness centrality value means a node is closer to more vertices in the graph, and thus more important or central in the graph.

In the problem of \CRI{}, we are given a graph $G$ and vertices $a,b$, and the goal is to add the $k$ edges to $G$ which will maximize the ratio of $a$ and $b$'s closeness centralities. We define the closeness ratio of $a$ and $b$ as $\ccrg{G}{a}{b}=\frac{\min(~\ccg{G+S}{a}~,~\ccg{G+S}{b}~)}{\max(~\ccg{G+S}{a}~,~\ccg{G+S}{b}~)}$, so that we can guarantee our ratio is never more than 1.
\optproblemdef{\CRI{}}{A graph $G=(V,E)$, vertices $a,b\in V$, and a positive integer $k\in\mathbb{N}$.}{Find a set of edges $S$ of size at most $k$ which maximizes $\ccrg{G+S}{a}{b}$}

A key aspect of the problem definition is the choice to consider the ratio of closeness centralities instead of their difference. Another natural formulation is the following:
\optproblemdef{\CGM}{A graph $G=(V,E)$, vertices $a,b\in V$, and a positive integer $k\in\mathbb{N}$.}{Find a set of edges $S$ of size at most $k$ which minimizes $|\ccg{G+S}{a}-\ccg{G+S}{b}|$.}
Within the fairness literature, both the ratio and difference are commonly considered variants for equalization of various measures \cite{friedler2019comparative}, but recent work supports the use of ratio by default unless domain-specific concerns indicate otherwise \cite{yeh2024analyzing}. In the context of closeness centrality equalization, we additionally find justification for the use of the ratio instead of the difference as, unlike \CRI, \CGM{} is unlikely to admit any multiplicative approximation algorithm, as we show in Section~\ref{sec:hardness-section} (\Cref{cor:gap-minimization-hard}).

We also propose the problem of adding edges to a graph to maximize the minimum pairwise closeness ratio. A ratio close to 1 indicates that all nodes are similarly advantaged. 
For a vertex subset $X \subseteq V$, we write $\bestcc{G}{X} \coloneqq \min_{x \in X} \ccg{G}{x}$ and $\worstcc{G}{X} \coloneqq \max_{x \in X} \ccg{G}{x}$. Formally, the problem is as follows.

\optproblemdef{\APCRI{}}{A graph $G=(V,E)$ and $k\in\mathbb{N}$.}{Find $S \subseteq V^2\setminus E$, $|S| \leq k$, maximizing $\frac{\bestcc{G+S}{V}}{\worstcc{G+S}{V}}$.}

In~\Cref{sec:hardness-section}, we establish that~\APCRI{} is \cclass{NP}-hard.

\subsection{Intuition-Building Examples}\label{appendix:strategies}
We examine three natural strategies for solving \CRI{} and demonstrate that none of them yield good approximations. 
Observe~\Cref{fig:intuition-examples}, and note that for every constant $d$, for sufficiently large $X$, the distances from $a$ and $b$ to vertices in $IS_X$ dominate their closeness centralities. Hence, throughout the section we will
simplify calculations by ignoring other asymptotically negligible distances.

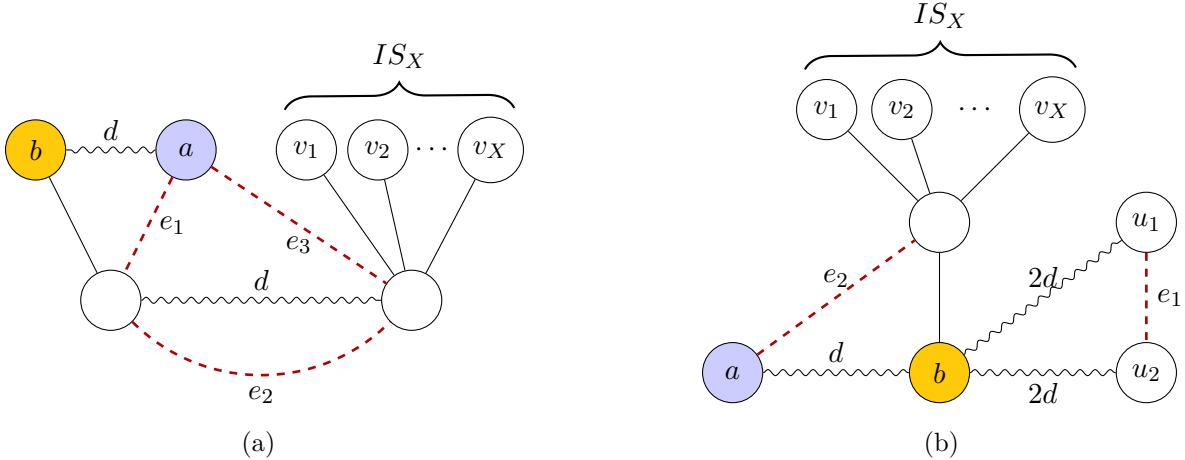
\begin{figure}
    \centering

    \begin{subfigure}[t]{0.45\textwidth}
        \centering
        \begin{tikzpicture}[ 
            roundnode/.style={circle, draw=black, minimum size=8mm},
            specialnode/.style={circle, draw=black, minimum size=18mm},
            rectnode/.style={rectangle, draw=black, minimum size=8mm}
        ]
        \node[roundnode, fill=blue!20] (A) at (3,1) {$a$};
        \node[roundnode, fill=yellow!60!orange] (B) at (1,1) {$b$};
        \node[roundnode] (C) at (2,-1) {};
        \node[roundnode] (D) at (6,-1) {};
    
        \node[roundnode] (IS1) at (4.6,1) {$v_1$};
        \node[roundnode] (IS2) at (5.55,1) {$v_2$};
        \node (dotsE) at (6.3,1) {$\cdots$};
        \node[roundnode] (ISX) at (7.05,1) {$v_X$};
        
        \draw[decorate, decoration={snake, amplitude=.4mm, segment length=2mm}] (A) -- (B) node[midway, above] {$d$};
        \draw[-] (B) -- (C);
        \draw[dashed, draw=red!70!black, line width=1pt] (A) -- (C) node[midway, right, text=black] {$e_1$};
        \draw[dashed, decorate, draw=red!70!black, line width=1pt, bend right=45] 
            (C) to node[midway, below, text=black] {$e_2$} (D);
        \draw[decorate, decoration={snake, amplitude=.4mm, segment length=2mm}] (C) -- (D) node[midway, above] {$d$};
    
        \draw[dashed, decorate, draw=red!70!black, line width=1pt] (A) to node[midway, below, text=black] {$e_3$} (D);
    
        \draw[-] (D) -- (IS1);
        \draw[-] (D) -- (IS2);
        \draw[-] (D) -- (ISX);
    
        \draw[decorate, decoration={brace, amplitude=10pt}, thick]
        ([yshift=7pt]IS1.north west) -- ([yshift=7pt]ISX.north east) node[midway, yshift=20pt] {$IS_X$};
    
        \end{tikzpicture}
        \caption{\centering\label{fig:intuition-examples-a}}
    \end{subfigure}
    \hfill
    \begin{subfigure}[t]{0.45\textwidth}
        \centering
        \begin{tikzpicture}[ 
            roundnode/.style={circle, draw=black, minimum size=8mm},
            specialnode/.style={circle, draw=black, minimum size=18mm},
        ]
        
        \node[roundnode, fill=blue!20] (A) at (-.75,1) {$a$};
        \node[roundnode, fill=yellow!60!orange] (B) at (2,1) {$b$};

        \node[roundnode] (C) at (2,3) {};

        \node[roundnode] (IS1) at (.5,4.5) {$v_1$};
        \node[roundnode] (IS2) at (1.5,4.5) {$v_2$};
        \node (dotsE) at (2.5,4.5) {$\cdots$};
        \node[roundnode] (ISX) at (3.5,4.5) {$v_X$};

        \node[roundnode] (U1) at (4.75,3) {$u_1$};
        \node[roundnode] (U2) at (4.75,1) {$u_2$};

        \draw[dashed, draw=red!70!black, line width=1pt] (U1) -- (U2) node[midway, right] {$e_1$};
        \draw[dashed, draw=red!70!black, line width=1pt] (A) -- (C) node[midway, above] {$e_2$};

        \draw[decorate, decoration={snake, amplitude=.4mm, segment length=2mm, post length=0mm, pre length=0mm}] (A) to node[midway, above] {$d$} (B);
        \draw[decorate, decoration={snake, amplitude=.4mm, segment length=2mm, post length=0mm, pre length=0mm}] (B) to node[midway, above] {$2d$} (U1);
        \draw[decorate, decoration={snake, amplitude=.4mm, segment length=2mm, post length=0mm, pre length=0mm}] (B) to node[midway, below] {$2d$} (U2);

        \draw[-] (B) -- (C);

        \draw[-] (C) -- (IS1);
        \draw[-] (C) -- (IS2);
        \draw[-] (C) -- (ISX);

        \draw[decorate, decoration={brace, amplitude=10pt}, thick]
        ([yshift=7pt]IS1.north west) -- ([yshift=7pt]ISX.north east) node[midway, yshift=20pt] {$IS_X$};
        
        \end{tikzpicture}
        \caption{\centering\label{fig:intuition-examples-b}}
    \end{subfigure}
    \caption{\label{fig:intuition-examples}The constructions used for counterexamples in Section~\ref{appendix:strategies}. Dashed lines are potential edge additions referred to in the text, squiggly lines are paths of length $d$ or $2d$, and $IS_X$ denotes an independent set of size $X$.}

\end{figure}

Much of the existing work on centrality maximization relies on centrality objectives being monotone and submodular \cite{crescenzi2016,bergamini2016,medya2018,shan2018}, but our objective satisfies neither property.
Observe that in~\Cref{fig:intuition-examples-a}, $\ccrg{G}{a}{b}=\frac{d+2}{2d+2}$.
Consider candidate edge additions $e_1$ and $e_2$ and the resulting ratios: $\ccrg{G+e_1}{a}{b}=1$, $\ccrg{G+e_2}{a}{b}=\frac{3}{d+3}$, and $\ccrg{G+e_1+e_2}{a}{b}=1$.
Certainly $\ccrg{G}{a}{b}>\ccrg{G+e_2}{a}{b}$ shows that our objective is not monotone.
We also have that $\ccrg{G+e_1}{a}{b}-\ccrg{G}{a}{b}=1-\frac{d+2}{2d+2}=\frac{d}{2d+2}$, while $\ccrg{G+e_2+e_1}{a}{b}-\ccrg{G+e_2}{a}{b}=1-\frac{3}{d+3}=\frac{d}{d+3}$.
The marginal gain from adding $e_1$ to the graph is larger after adding $e_2$, and thus our objective is not submodular.

We also observe that the task of equalizing the centrality of two vertices is not necessarily the same as greedily improving the less central of the two (until they are equal). Again using~\Cref{fig:intuition-examples-a}, given that $\ccg{G}{a} > \ccg{G}{b}$, if we want to improve $\ccg{G}{a}$ as much as possible with $k=1$, we would add $e_3$.
But $\ccrg{G+e_3}{a}{b} = \frac{2}{d+2}$ while $\ccrg{G+e_1}{a}{b} = 1$, so this strategy does not lead to any constant-factor approximation.

Finally, we illustrate how our problem diverges from the well-researched problem of adding edges to a graph to minimize its diameter \cite{demaine,li,adriaens2022diameter,bilo,frati2015augmenting}.
In~\Cref{fig:intuition-examples-b}, we claim that $e_1$ is a diameter-minimizing edge addition.
Suppose otherwise; let $pq$ be an edge addition achieving diameter strictly less than $3d$, and WLOG assume $\distg{G}{a}{p} \leq \distg{G}{a}{q}$. If $\distg{G}{a}{p} \geq d$, then $\distg{G}{u_1}{q} < 2d$ and $\distg{G}{u_2}{q} < 2d$, but no choice of $q$ satisfies both inequalities, so we may assume that $\distg{G}{a}{p} = d - x$ for some positive $x$. Now, WLOG assume that $q$ lies on the path from $b$ to $u_2$, and let $\distg{G}{u_2}{q} = y$. We may assume that $x + y < d$, as otherwise $\distg{G+pq}{u_1}{u_2} \geq 3d$. But then it is simple to check that $\distg{G+pq}{a}{u_1} = 3d$. So, $e_1$ is diameter-minimizing, but $\ccrg{G+e_1}{a}{b} = \frac{2}{d+2}$, while $\ccrg{G+e_2}{a}{b} = 1$. It follows that minimizing the graph's diameter does not lead to any constant-factor approximation for~\CRI{}.

\subsection{A Trivial $\frac{1}{2}$-Approximation}\label{sec:half-approx}
If $ab \in E$, then for every vertex $v$ we have $|\distg{G}{a}{v} - \distg{G}{b}{v}| \leq 1$. This also yields a bound on the ratio of centralities.
    
    
    
    
    
\begin{restatable}{observation}{adjacentobservation} \label{observation:adjacent}
    If $a$ and $b$ are adjacent, then $\ccrg{G}{a}{b} > \frac{1}{2}$.
\end{restatable}
\begin{proof}
    Assume WLOG that $\ccg{G}{b}\leq\ccg{G}{a}$. By the triangle inequality, if $a$ and $b$ are adjacent, then for any vertex $v$, $\distg{G}{a}{v}\leq \distg{G}{a}{b}+\distg{G}{b}{v}=1+\distg{G}{b}{v}$. It follows that
    \begin{align*}
        \ccg{G}{a} &= \distg{G}{a}{a} + \distg{G}{a}{b} + \sum_{v \in V\setminus\{a, b\}} \distg{G}{a}{v} \\
        &= \distg{G}{b}{b} + \distg{G}{a}{b} + \sum_{v \in V\setminus\{a, b\}} \distg{G}{a}{v} \\
        &\leq \distg{G}{b}{b} + \distg{G}{a}{b} + \sum_{v \in V\setminus\{a, b\}} (1 + \distg{G}{b}{v})
        = \ccg{G}{b} + n - 2.
    \end{align*}
    Furthermore, it is trivially true that $\ccg{G}{b} \geq n - 1$. Thus, we may write $\ccg{G}{b} = n - 1 + x$ for some non-negative $x$.
    Then we have
    \[
        \ccrg{G}{a}{b} = \frac{\ccg{G}{b}}{\ccg{G}{a}} \geq \frac{\ccg{G}{b}}{\ccg{G}{b} + (n - 2)} = \frac{(n - 1) + x}{(n - 1) + x + (n - 2)} \geq \frac{n - 1}{2(n - 1) - 1} > \frac{1}{2},
    \]
    where in the second to last inequality we have made use of a fundamental observation:
    \begin{observation}\label{obs:adding-to-fractions}
        For all real numbers $w, p, q$ with $0 \leq w$ and $0 < p \leq q$, $\frac{p + w}{q + w} \geq \frac{p}{q}$.    
    \end{observation}
    \Cref{obs:adding-to-fractions} is simple to verify via a proof by contradiction.
\end{proof}

Recalling that no ratio greater than $1$ is possible, a trivial~$\frac{1}{2}$-approximation for \CRI{} follows: add the edge $ab$ if it is not already present. This is also \emph{no better} than a $\frac{1}{2}$-approximation. Consider a star centered on $b$, with $a$ a leaf. Given $k \geq n-2$, an optimal algorithm achieves a ratio of $1$, but the edge $ab$ alone achieves only $\frac{1}{2} + o(1)$.
In~\Cref{sec:approximation-alg}, we will show how to improve this approximation factor. In~\Cref{sec:hardness-section}, we will prove lower bounds.

\section{Lower Bounds} \label{sec:hardness-section}

In this section, we establish lower bounds. Full proofs can be found in~\Cref{appendix:hardness}.
First, we formally define the decision variant of \CRI{}. 
\problemdef{\CRI{} (Decision Variant)}{A graph $G=(V,E)$, vertices $a,b\in V$, a budget $k\in\mathbb{N}$, and a target ratio $\tau\in(0,1]$.}
{Does there exist a set $T \subseteq V^2$ of size at most $k$ such that $\ccrg{G+T}{a}{b}\geq\tau$?}

We begin by giving a polynomial and parameter-preserving reduction from \SC{} when $\tau = 1$, implying hardness in a very restricted regime: it is not even easy to tell whether perfect equality can be achieved.

\begin{restatable}{theorem}{nphardtone} \label{t1-hard}
    \CRI{} is \cclass{NP}-hard and \cclass{W[2]}-hard\footnote{\cclass{W[2]}-hardness can be interpreted as very strong evidence that no fixed-parameter tractable algorithm exists for the parameter $k$. For an introduction to parameterized complexity, see the standard text~\cite{cygan2015parameterized}.} with respect to $k$, even when $\tau = 1$.
\end{restatable}

\begin{proof}[Proof sketch] 
    Let $(U, S_1, S_2, \ldots, S_m, k)$ be an instance of \SC{}, where $U$ denotes a universe of $n$ elements and $\{S_j\}_{j = 1}^m$ is a set family over $U$.
    We construct an instance $(G = (V, E), a, b, k', \tau = 1)$ of \CRI{} as follows; see~\Cref{cri-instance} for a visual depiction.
    We create three vertices $a, b, c$, which induce a triangle. We create $n$ \emph{element vertices} $v_1, v_2, \ldots, v_n$ corresponding to the elements $u_1, u_2, \ldots, u_n$ of $U$, ordered arbitrarily.
    For each set $S_j$, we create a \emph{set vertex} $s_j$, as well as the edge $cs_j$ and the edges $s_jv_i$ for all $u_i \in S_j$. 
    Finally, we create $X = n + k$ vertices, and add an edge from each of these vertices to $a$. 
    We also set $k' = k$. This completes the construction.

    \begin{figure}
        \centering
        \begin{tikzpicture}[ 
            roundnode/.style={circle, draw=black, fill =white, minimum size=8mm},
            goldnode/.style={circle, draw=black, fill=red, minimum size=8mm}
            bluenode/.style={circle, draw=black, fill=red, minimum size=8mm},
        ]
    
        \node[roundnode, fill=blue!20] (a) at (-1.5,3) {$a$};
        \node[roundnode, fill=yellow!60!orange] (b) at (1.5,3) {$b$};
        \node[roundnode] (c) at (0,1) {$c$};
        \node[roundnode] (IS1) at (-3.5,4) {$z_1$};
        \node[roundnode] (IS2) at (-3.5,3) {$z_2$};
        \node (dotsIS) at (-3.5,2.35) {$\vdots$};
        \node[roundnode] (ISX) at (-3.5,1.5) {$z_X$};
        \node[roundnode] (S1) at (-2,-0.75) {$s_1$};
        \node[roundnode] (S2) at (-.5,-0.75) {$s_2$};
        \node (dotsS) at (.75,-0.75) {$\cdots$};
        \node[roundnode] (Sm) at (2,-0.75) {$s_m$};
        \node[roundnode] (E1) at (-3,-3) {$v_1$};
        \node[roundnode] (E2) at (-1.5,-3) {$v_2$};
        \node[roundnode] (E3) at (0,-3) {$v_3$};
        \node (dotsE) at (1.5,-3) {$\cdots$};
        \node[roundnode] (En) at (3,-3) {$v_n$};
        \draw[-] (a) -- (b);
        \draw[-] (a) -- (c);
        \draw[-] (b) -- (c);
        \draw[-] (c) -- (S1);
        \draw[-] (c) -- (S2);
        \draw[-] (c) -- (Sm);
        \draw[-] (S1) -- (E1);
        \draw[-] (S1) -- (E2);
        \draw[-] (Sm) -- (E3);
        \draw[-] (S2) -- (E2);
        \draw[-] (Sm) -- (E1);
        \draw[-] (Sm) -- (En);
        \draw[-] (a) -- (IS1);
        \draw[-] (a) -- (IS2);
        \draw[-] (a) -- (ISX);
        \draw[decorate, decoration={brace, amplitude=10pt}, thick]
        ([xshift=-7pt]ISX.south west) -- ([xshift=-7pt]IS1.north west)
        node[midway, xshift=-25pt] {$IS_X$};

        \end{tikzpicture}
        \caption{The construction given by~\Cref{t1-hard}. We denote set vertices by $s_j$ and element vertices by $v_i$. Here, $IS_X$ is an independent set of $X$ vertices, with each vertex adjacent to $a$. We refer to the proof of~\Cref{t1-hard} for a formal description of the construction and the accompanying analysis.}
        \label{cri-instance}
        \end{figure}

    We claim that there exists a set cover of $U$ of size at most $k$ if and only if there is a set $T$ of at most $k$ edges such that $\ccrg{G+T}{a}{b}=1$.
    For the forward direction, suppose that there exists a cover $C \subseteq \{S_j\}_{j=1}^m$ of size $k$. Let $T = \{bs_j \colon S_j \in C\}$, i.e., the set of edges from $b$ to the set vertices corresponding to sets in the cover $C$.
    Then for each element vertex $v_i$, $\distg{G+T}{v_i}{b} = 2$, and $\distg{G+T}{v_i}{a} = 3$.
    Also, $b$ is adjacent to $k$ set vertices, so $\sum_{j = 1}^m \distg{G+T}{s_j}{b} = 2m - k$, while $\sum_{j = 1}^m \distg{G+T}{s_j}{a} = 2m$.
    Then
    \[\ccg{G+T}{b} = 2X + 1 + 1 + 2m - k + 2n = 2(n+k) + 2m - k + 2n + 2 = 4n + 2m + k + 2.\]
    Similarly,
      $\ccg{G+T}{a} = X + 1 + 1 + 2m + 3n = 4n + 2m + k + 2$,  
    so $\ccg{G+T}{b} = \ccg{G+T}{a}$, implying $\ccrg{G+T}{a}{b} = 1 = \tau$, as desired.

    For the reverse direction, assume that there is no collection of $k$ sets which covers $U$.
    We must prove that there is no set $T \subseteq V^2 \setminus E$ of size at most $k$ such that $\ccg{G+T}{b} = \ccg{G+T}{a}$.
    We accomplish this by showing that for every set $T$ of size at most $k$,  $\ccg{G+T}{b} > \ccg{G}{a}$.
    This is sufficient to complete the proof, since it follows from monotonicity that $\ccg{G+T}{a} \leq \ccg{G}{a}$.
    Toward this end, let $T^*$ be a set of (at most) $k$ edge additions such that for all sets $T \subseteq V^2 \setminus E$ of size at most $k$, $\ccg{G+T^*}{b} \leq \ccg{G+T}{b}$.
    We now state a useful structural lemma, which will allow us to assume that all edges in $T^*$ connect $b$ to set vertices.

    \begin{restatable}{claim}{setvertexincidentclaim} \label{claim:star-on-b}
        There exists a set $T \subseteq V^2 \setminus E$ of size at most $k$ such that every edge in $T$ is incident on $b$, all other endpoints of edges in $T$ are set vertices, and $\ccg{G+T}{b} \leq \ccg{G+T^*}{b}$.
    \end{restatable}
    The condition that all edges are incident on $b$ follows from an exchange argument, making use of the fact that (by the triangle inequality) an edge $xy$ appears along shortest paths from $b$ in at most one orientation, i.e., either $x$ before $y$ or $y$ before $x$. The condition that all other endpoints are set vertices then follows from a second exchange argument, this time leveraging the particular structure of our construction. See~\Cref{appendix:hardness-star} for the full proof.

    \Cref{claim:star-on-b} allows us to give a useful lower bound on $\ccg{G+T^*}{b}$.
    We use the fact that the sets associated with the endpoints of edges in $T^*$ are (by hypothesis) \emph{not} a cover of $U$ to
    conclude that there exists at least one element vertex $v_i$ with $\distg{G+T^*}{b}{v_i} = 3$.
    Then,
    \begin{align*}
        \ccg{G+T^*}{b} &\geq 2X + 2 + 2m - k + 2(n - 1) + 3 = 2(n+k) + 2 + 2m - k + 2(n-1) + 3 \\
        &= 4n + 2m + k + 3
        > 4n + 2m + k + 2 = \ccg{G}{a},        
    \end{align*}
    as desired.
\end{proof}
\noindent
Observe that a target ratio of $1$ for \CRI{} corresponds exactly to a target gap of $0$ in the context of \CGM{}.
Thus,~\Cref{t1-hard} also (conditionally) rules out any multiplicative approximation for the latter problem.

\begin{restatable}{corollary}{gapminimizationhard}\label{cor:gap-minimization-hard}
    \CGM{} is \cclass{NP}-hard and \cclass{W[2]}-hard with respect to $k$, and (unless \cclass{P} = \cclass{NP}) admits no polynomial-time multiplicative approximation algorithm.
\end{restatable}

In \Cref{sec:half-approx}, we observed that \CRI{} is trivial when $\tau \leq \frac{1}{2}$.
Meanwhile, by~\Cref{t1-hard}, the problem is hard when $\tau = 1$.
The natural question is whether we can find the barrier between tractability and \cclass{NP}-hardness.
We resolve this question by showing that \CRI{} is \cclass{NP}-hard for every $\tau > \frac{1}{2}$.
The construction is a generalization of that given by~\Cref{t1-hard}; we use the value of $\tau$ to
determine an appropriate size for the independent set adjacent to $a$, as well as to determine an appropriate number of vertices corresponding to each element in the \SC{} instance. 

\begin{restatable}{theorem}{nphardtaulessthanone} \label{t2-hard}
    \CRI{} is \cclass{NP}-hard \cclass{W[2]}-hard with respect to $k$, even when restricted to any constant $\tau > \frac{1}{2}$.
\end{restatable}
    Let $\tau\in(\frac{1}{2},1)$, and note that~\Cref{t1-hard} covers the case where $\tau = 1$. Consider an instance of \SC{} with $m$ sets, $n$ elements, and $k\in\mathbb{N}$.
    We first claim that we may assume $\frac{2m+4n+k}{1+2m+4n+k}\geq\tau^2$. Otherwise, we create an equivalent instance of \SC{} with $\lceil\frac{\tau^2}{4n(1 - \tau^2)}\rceil$ copies of each element, giving each copy the same set memberships as the original element.\footnote{Conceptually, one may also think of this procedure as the introduction of $\lceil\frac{\tau^2}{4n(1 - \tau^2)}\rceil$ twins for each element vertex in our constructed \CRI{} instance.}
    Then, letting $n'$ denote the number of elements in our equivalent instance, we have
    \begin{align*}
        n' &\geq \frac{\tau^2}{4(1 - \tau^2)} \\
        4n'(1 - \tau^2) &\geq \tau^2 > \tau^2 + 2m(\tau^2 - 1) + k(\tau^2 - 1) \\
        \frac{2m + 4n' + k}{1 + 2m + 4n' + k} &\geq \tau^2,
    \end{align*}
    so our assumption is safe. We now repeat the same construction given by~\Cref{t1-hard} and depicted in~\Cref{cri-instance}.
    This time, we let the number $X$ of vertices in the independent set adjacent to $a$ be an integer in the interval given by the
    following lemma. We prove that a suitable value always exists.

    \begin{restatable}{claim}{integerintervalclaim}\label{clm:integer}
        For every $m,n,k\in\mathbb{N}$ and $\tau\in\left(\frac{1}{2},1\right)$, there exists a natural number $X$ in the interval $\left(\frac{2+2m+3n-\tau(2+2m-k+2n+1)}{2\tau-1},\frac{2+2m+3n-\tau(2+2m-k+2n)}{2\tau-1}\right]$.
    \end{restatable}
    \begin{proof}
        Taking the difference of the bounds of the interval, we get that the length of the interval is $\frac{\tau}{2\tau-1}$. If $\tau<1$, then $-\tau>-1$ and thus $\tau>2\tau-1$. Thus the length of the interval is greater than 1 when $\tau\in(\frac{1}{2},1)$, so there must exist some integer $X$ within this interval. Furthermore, for $m,n,k\in\mathbb{N}$ and $\tau\in(\frac{1}{2},1)$, the bounds of this interval are  both positive, so $X$ must be a natural number, as desired.
    \end{proof}

    If there exists a set cover $C$ of size $k$, then we add $k$ edges from $b$ to the set vertices corresponding to $C$.
    In the analysis, some care must be taken to handle the case in which these edge additions cause $b$'s closeness centrality to become lower than that of $a$.
    For the other direction, we will use our selected parameter values to show that for every set $T$ of $k$ edge additions,
    $\ccg{G+T}{b} > \frac{1}{\tau}\cdot\ccg{G}{a} \geq \frac{1}{\tau}\cdot\ccg{G+T}{a}$.
    We defer the arithmetic to~\Cref{appendix:hardness-generalized}.

\bigskip
Next, we establish a bicriteria inapproximability result for \CRI{}.
We make use of a result originally suggested by Feige's work on the inapproximability of \SC{}~\cite{feige1998threshold}, and now firmly established in the literature, e.g.,~\cite{cohen2021approximability}. Let $c \geq 1$ be a constant. Then given a family of $m$ sets over $n$ elements, it is \cclass{NP}-hard to distinguish between
the existence of a set cover of size $k$, and the non-existence of any collection of $ck$ sets which cover more than $(1 - (1 - \frac{1}{k})^{ck} + \delta)n$ elements, where $\delta > 0$ may be arbitrarily small.
In particular, when designing our reduction, we may assume that if there is no set cover of size $k$, then every collection of $ck$ sets covers at most $(1 - \frac{1}{e^c} + \delta)n$ elements.
The construction and analysis are similar to~\Cref{t1-hard}; see~\Cref{appendix:hardness-inapprox} for the proof.

\begin{restatable}{theorem}{approxhardness}\label{thm:inapprox}
    For every $\varepsilon > 0, c \geq 1$, unless \cclass{P} = \cclass{NP}, there is no polynomial-time algorithm which takes as input an instance $(G = (V, E), a, b, k)$ of \CRI{} with optimal objective value $\opt$ and returns a set $T \subseteq V^2\setminus E$ of size at most $ck$ while guaranteeing that $\ccrg{G+T}{a}{b} \geq \frac{5e^c}{5e^c + 1 - \varepsilon}\cdot\opt$.
\end{restatable}

Finally, we establish a lower bound for~\APCRI{}. The decision variant formally defined as follows.
\problemdef{\APCRI{} (\textsc{Decision Variant})}{A graph $G = (V,E)$ a budget $k \in \mathbb{N}$, and a target ratio $\tau \in (0,1]$.}{Does there exist a set $S \subseteq V^2 \setminus E$ of size at most $k$ such that $\frac{\bestcc{G+S}{V}}{\worstcc{G+S}{V}} \geq \tau$?}

Our polynomial and parameter-preserving reduction is from~\domset{}. See~\Cref{appendix:all-pairs} for the proof.

\begin{restatable}{theorem}{allpairshardness}\label{thm:all-pairs}
\APCRI{} is \cclass{NP}-hard and \cclass{W[2]}-hard with respect to $k$.
\end{restatable}

\section{A $\frac{6}{11}$-Approximation}\label{sec:approximation-alg}

In this section we show that a surprisingly simple strategy can improve upon the $\frac{1}{2}$-approximation yielded by adding the edge $ab$: we need only add edges from $a$ to $b$'s neighborhood.
Recall that we assume $\ccg{G}{b} < \ccg{G}{a}$. The main intuition is that $\ccrg{G+ab}{a}{b}$ is actually much larger than $\frac{1}{2}$ unless a significant portion of the graph is in $b$'s neighborhood but not $a$'s. Even in this case, unless the optimum objective value is very close to $1$, $G+ab$ is already~$\frac{6}{11}$-approximate.
These two insights impose significant structural restrictions on hard instances, and the achievement of this section is to show how to reason about these instances.
We do so via careful arguments to bound the sizes of relevant vertex sets. 
Generally, the intuition behind these arguments will be straightforward, and the challenge will be to appropriately formalize the necessary inequalities.

\textbf{Additional Notation:} Throughout the section, given a graph $G = (V, E)$ and vertices $a, b \in V$, we denote by $N_G(a)$ the open neighborhood of $a$ in $G$, $N_G[a] = N_G(a) \cup \{a\}$ the closed neighborhood of $a$, and $A^p = N_G(a) \setminus N_G[b]$ the \emph{private neighbors} of $a$.
Given a vertex subset $O \subseteq V$, we write $\meandistg{G}{a}{O} = \frac{1}{|O|}\cdot\sum_{v \in O} \distg{G}{a}{v}$ for the mean distance from $a$ to vertices in $O$.
We now state our algorithm in~\Cref{alg:neighborhood}.

\bigskip
    \begin{algorithm}[H]

      \SetKwFunction{Union}{Union}
      \SetKwFunction{FindCompress}{FindCompress}
      
    \Input{ An instance $(G = (V, E), a, b, k)$ of \CRI{}, with $\ccg{G}{b} < \ccg{G}{a}$.}
    \Output{ A set $S \subseteq V^2 \setminus E$ of cardinality at most $k$.}

        $S \gets \emptyset$ \\
        \If{$ab \notin E$}{
            $S \gets S \cup \{ab\}$ \\
            \If{$\ccg{G+S}{a} \leq \ccg{G+S}{b}$}{\label{algline:neighborhood-no-switching-condition}
                \Return $S \text{ or } \emptyset\text{, whichever is better.}$\label{algline:neighborhood-no-switching-return}
            }
        }
        \While{$|S| < k$}{\label{algline:neighborhood-while-loop}
            $u \gets \text{arbitrary vertex from } N_{G+S}(b) \setminus N_{G+S}[a]$ \\\label{algline:neighborhood-vertex-selection}
            $S \gets S \cup \{au\}$ \\
            \If{$\ccg{G+S}{a} \leq \ccg{G+S}{b}$}{\label{algline:neighborhood-termination}
            \Return $S \text{ or } S \setminus \{au\}\text{, whichever is better.}$\label{algline:neighborhood-termination-line}
            }
        }
        \Return $S$ or $\emptyset$, whichever is better.\label{algline:neighborhood-final-return}
      \caption{\label{alg:neighborhood}}
    \end{algorithm}

\bigskip

Before analyzing~\Cref{alg:neighborhood}, we remark on the running time. Observe that the \textbf{while} loop iterates at most $k$ times, since the cardinality of $S$ increases with every iteration. Then, using breadth-first search for centrality computations, a naive implementation of~\Cref{alg:neighborhood} runs in $O(k(n+m))$ time. We now improve this running time to quasilinear: $O((n+m)\log k)$.
The key insight is that as the \textbf{while} loop iterates, $\ccg{G+S}{a}$ monotonically decreases while $\ccg{G+S}{b}$ stays constant.
To see that the centrality of $b$ stays constant, observe that once we have reached the loop, $a \in N_{G+S}(b)$. Then the edge $au$ under consideration in any given iteration has both endpoints in $N_{G+S}(b)$, and therefore cannot appear on any shortest path from $b$.
Thus, to improve the running time, instead of evaluating the condition on line~\ref{algline:neighborhood-termination} in each iteration of the loop we may check at the end of the algorithm whether $\ccg{G+S}{a} \leq \ccg{G+S}{b}$. If so, then we determine the correct return value via a binary search over the $k-1$ previous states of $S$.

\subsection{Analysis I: The Early Termination Conditions}\label{sec:termination-conditions}
To analyze~\Cref{alg:neighborhood}, we begin by proving that the pair of early termination conditions on lines~\ref{algline:neighborhood-no-switching-condition} and~\ref{algline:neighborhood-termination} are safe.
Recall that we assume $\ccg{G}{b} < \ccg{G}{a}$, WLOG. 
It will be very helpful in our analysis to assume that the centrality of $b$ remains less than that of $a$ throughout the algorithm. 
Toward this end, we show in~\Cref{lemma:no-switching} that either
$\ccg{G+ab}{b} < \ccg{G+ab}{a}$ or one of $G$ or $G + ab$ satisfies our desired approximation bound.
In other words, we may assume that the more central vertex does not ``switch'' when the edge $ab$ is added, and thus the condition given on line~\ref{algline:neighborhood-no-switching-condition} is safe.

\begin{lemma}\label{lemma:no-switching}
    Let $(G, a, b, k)$ be an instance of~\CRI{} with $\ccg{G}{b} < \ccg{G}{a}$. Then at least one of the following is true:
    \begin{enumerate}[(i)]
        \item $\ccg{G+ab}{b} < \ccg{G+ab}{a}$,
        \item $\ccrg{G+ab}{a}{b} \geq \frac{6}{11}$, or
        \item $\ccrg{G}{a}{b} \geq \frac{6}{11}$.
    \end{enumerate}
    \end{lemma}
\begin{proof}

    Assume toward a contradiction that none of (i), (ii), or (iii) is true. 
    Recall that $A^p = N_G(a) \setminus N_G[b] = N_{G+ab}(a) \setminus N_{G+ab}[b]$ denotes the private neighbors of $a$.
    Let $M = N_G(a) \cap N_G(b)$ denote the mutual neighbors of $a$ and $b$.
    Let $R = V(G) \setminus (A^p \cup M \cup \{a, b\})$.
    We will begin by using the assumption that parts (i) and (ii) do not hold, i.e., $\ccg{G+ab}{a} < \frac{6}{11}\cdot\ccg{G+ab}{b}$, to show that $A^p$ must be large.

    \begin{claim}\label{claim:switching-large-set}
        $|A^p| > \frac{4(n - 2)}{5}$.
    \end{claim}
    To understand the intuition behind~\Cref{claim:switching-large-set}, observe that for any vertex $v$, $\distg{G+ab}{a}{v}$ and $\distg{G+ab}{b}{v}$ differ by at most one.
    Hence, for every $v$, the ratio $\distg{G+ab}{a}{v}/\distg{G+ab}{b}{v} \geq \frac{1}{2}$, and moreover if $v$ is not a private neighbor of $a$ then $\distg{G+ab}{a}{v}/\distg{G+ab}{b}{v} \geq \frac{2}{3}$. 
    But under the assumption that $\ccg{G+ab}{a} < \frac{6}{11}\cdot\ccg{G+ab}{b}$, this means that a large portion of the graph must be in $a$'s private neighborhood.
    We include the details here so that the reader may see the style of the argument; similar arguments will be used throughout the analysis, and in the future we will defer the details to~\Cref{appendix:early-termination-lemmas}.
    \begin{proof}[Proof of~\Cref{claim:switching-large-set}]
        Assume otherwise. Then there exist non-negative $x_R, x_M, y_R, y_M$ such that $y_R + y_M = \frac{(n - 2)}{5}$, $|A^p| = \frac{4(n - 2)}{5} - x_R - x_M$, $|R| = y_R + x_R$, and $|M| = y_M + x_M$.
        Now we complete the proof by comparing $\ccg{G+ab}{a}$ to an upper bound for $\ccg{G+ab}{b}$.
        We begin with the former:
        \begin{align*}
            \ccg{G+ab}{a} &= \distg{G+ab}{a}{a} + \distg{G+ab}{a}{b} + 1\cdot|A^P| + 1\cdot|M| + \meandistg{G+ab}{a}{R}\cdot|R| \\
            &= 1 + \frac{4(n-2)}{5} - x_R - x_M + y_M + x_M + \meandistg{G+ab}{a}{R}\cdot(y_R + x_R).
        \end{align*} 
        Observe that by definition of $R$ we have $R \cap N_{G+ab}[a] = \emptyset$, so we may write $\meandistg{G+ab}{a}{R} = 2 + z$ for some non-negative $z$.
        Then we have
        \begin{align*}
            \ccg{G+ab}{a} &= 1 + \frac{4(n-2)}{5} - x_R - x_M + y_M + x_M + 2y_R + 2x_R + zy_R + zx_R \\
            &= 1 + \frac{4(n-2)}{5} + y_M + 2y_R + x_R + zy_R + zx_R \\
            &= 1 + \frac{6(n-2)}{5} - y_M + x_R + zy_R + zx_R \quad\quad(\text{using } y_R = \frac{n-2}{5} - y_M).
        \end{align*}
        Now we write an upper bound for $\ccg{G+ab}{b}$, using that for every $v \in A^p$, $\distg{G+ab}{b}{v} = 2$.
        \begin{align*}
            \ccg{G+ab}{b} &= \distg{G+ab}{b}{b} + \distg{G+ab}{b}{a} + 2\cdot|A^P| + 1\cdot|M| + \meandistg{G+ab}{b}{R}\cdot|R| \\
            &= 1 + \frac{8(n-2)}{5} - 2x_M - 2x_R + y_M + x_M + \meandistg{G+ab}{b}{R}(y_R + x_R).
        \end{align*}
        By the triangle inequality, $\meandistg{G+ab}{b}{R} \leq \meandistg{G+ab}{a}{R} + 1 = 3 + z$, so
        \begin{align*}
            \ccg{G+ab}{b} &\leq 1 + \frac{8(n-2)}{5} - 2x_M - 2x_R + y_M + x_M + 3y_R + 3x_R + zy_R + zx_R \\
            &= 1 + \frac{8(n-2)}{5} + y_M + 3y_R + x_R + zy_R + zx_R - x_M \\
            &\leq 1 + \frac{8(n-2)}{5} + y_M + 3y_R + x_R + zy_R + zx_R \\
            &= 1 + \frac{11(n-2)}{5} -2y_M + x_R + zy_R + zx_R \quad\quad(\text{using } y_R = \frac{n-2}{5} - y_M).
        \end{align*}
        Now we will examine $\ccrg{G+ab}{a}{b}$.
        We note that since $y_M \leq \frac{n-2}{5}$, we have $\frac{6(n-2)}{5} - y_M < \frac{11(n-2)}{5} -2y_M$.
        Thus, we may apply~\Cref{obs:adding-to-fractions} as follows.
        \begin{align*}
            \frac{\ccg{G+ab}{a}}{\ccg{G+ab}{b}}
            \geq \frac{1 + \frac{6(n-2)}{5} - y_M + x_R + zy_R + zx_R}{1 + \frac{11(n-2)}{5} -2y_M + x_R + zy_R + zx_R}
            \geq \frac{\frac{6(n-2)}{5} - y_M}{\frac{11(n-2)}{5} -2y_M}.
        \end{align*}
        Now we wish to eliminate the remaining $y_M$ terms. We use a fundamental observation:
        \begin{observation}\label{obs:subtracting-from-fractions}
            For all real numbers $w, q, p$ with $0 \leq w$ and $0 < \frac{q}{2} \leq p \leq q$, $\frac{p - w}{q - 2w} \geq \frac{p}{q}$.
        \end{observation}
            \Cref{obs:subtracting-from-fractions} is simple to verify with a proof by contradiction. We use this observation to complete the proof:
        \begin{align*}
            \ccrg{G+ab}{a}{b} &= \frac{\ccg{G+ab}{a}}{\ccg{G+ab}{b}} \geq \frac{\frac{6(n-2)}{5} - y_M}{\frac{11(n-2)}{5} -2y_M} \geq \frac{\frac{6(n-2)}{5}}{\frac{11(n-2)}{5}} = \frac{6}{11},
        \end{align*}
        contradicting the assumption that part (ii) of the lemma does not hold.
    \end{proof}

    Now we will use~\Cref{claim:switching-large-set} and the assumption that part (iii) does not hold, i.e., $\ccg{G}{b} < \frac{6}{11}\cdot\ccg{G}{a}$, to show that $\distg{G}{a}{b} = 1$.
    This will complete the proof of~\Cref{lemma:no-switching}, as then $G+ab = G$, so by hypothesis $\ccg{G+ab}{b} = \ccg{G}{b} < \ccg{G}{a} = \ccg{G+ab}{a}$, which is an absurdity since we have already used the assumption that part (i) does not hold to obtain~\Cref{claim:switching-large-set}.
    
    In the following, let $O = M \cup R$, i.e, all of the vertices besides $a$, $b$, and $a$'s private neighbors. 
    We will begin by giving a lower bound for $\ccg{G}{b}$. Using the triangle inequality, we note that for each $v \in A^p$, $\distg{G}{b}{v} \geq \distg{G}{a}{b} - 1$.
    Then we have
    \begin{align*}
        \ccg{G}{b} &= \distg{G}{b}{b} + \distg{G}{b}{a} + \meandistg{G}{b}{A^p}\cdot|A^p| + \meandistg{G}{b}{O}\cdot|O| \\
        &\geq \distg{G}{b}{a} + (\distg{G}{a}{b} - 1)\cdot|A^p| + \meandistg{G}{b}{O}\cdot|O|
    \end{align*}
    Now we write an upper bound for $\ccg{G}{a}$. We again make use of the triangle inequality: for each $v \in O$,~$\distg{G}{a}{v} \leq \distg{G}{a}{b} + \distg{G}{b}{v}$. Then $\meandistg{G}{a}{O} \leq \distg{G}{a}{b} + \meandistg{G}{b}{O}$, and thus
    \begin{align*}
        \ccg{G}{a} &= \distg{G}{a}{a} + \distg{G}{a}{b} + |A^p| + \meandistg{G}{a}{O}\cdot|O| \\
        &\leq \distg{G}{a}{b} + |A^p| + (\distg{G}{a}{b} + \meandistg{G}{b}{O})\cdot|O|.
    \end{align*} 
    From the assumption that part (iii) does not hold, we now have the following,
    \begin{align*}
        \ccg{G}{b} &< \frac{6}{11}\cdot\ccg{G}{a} \\
        \distg{G}{a}{b} + (\distg{G}{a}{b} - 1)\cdot|A^p| + \meandistg{G}{b}{O}\cdot|O| &< \frac{6}{11}\bigg[\distg{G}{a}{b} + |A^p| + (\distg{G}{a}{b} + \meandistg{G}{b}{O})\cdot|O|\bigg] \\
        \frac{5}{11}\cdot\distg{G}{a}{b} + \left(\distg{G}{a}{b} - \frac{17}{11}\right)\cdot|A^p| &< \bigg[\frac{6}{11}\distg{G}{a}{b} - \frac{5}{11}\meandistg{G}{b}{O}\bigg]\cdot|O| \\
        \left(\distg{G}{a}{b} - \frac{17}{11}\right)\cdot|A^p| &< \bigg[\frac{6}{11}\distg{G}{a}{b} - \frac{5}{11}\meandistg{G}{b}{O}\bigg]\cdot|O|. 
    \end{align*}

    Recall that if $\distg{G}{a}{b} = 1$ then we are done, since this is what we are trying to prove.
Then we may assume that $\distg{G}{a}{b} \geq 2$, which implies that the LHS of the above inequality is positive.
It follows that the RHS is also positive. We now use~\Cref{claim:switching-large-set} to substitute $\frac{4(n-2)}{5} < |A^p|$ and $|O| < \frac{n-2}{5}$.
\begin{align*}
    \left(\distg{G}{a}{b} - \frac{17}{11}\right)\cdot\frac{4(n-2)}{5} &< \bigg[\frac{6}{11}\distg{G}{a}{b} - \frac{5}{11}\meandistg{G}{b}{O}\bigg]\cdot\frac{n-2}{5} \\
    \distg{G}{a}{b}\cdot\frac{44(n - 2)}{55} - \frac{68(n-2)}{55} &< \distg{G}{a}{b}\cdot\frac{6(n-2)}{55} - \meandistg{G}{b}{O}\cdot\frac{5(n-2)}{55} \\
    \distg{G}{a}{b}\cdot\frac{38(n - 2)}{55} &< \frac{68(n-2)}{55} - \meandistg{G}{b}{O}\frac{5(n-2)}{55} \\
    \distg{G}{a}{b}\cdot\frac{38(n - 2)}{55} &< \frac{68(n-2)}{55} \\
    \distg{G}{a}{b} &< \frac{68}{38} < 2,
\end{align*} 
so $\distg{G}{a}{b} = 1$, contradicting the assumption that part (i) does not hold, as desired.
\end{proof}

Next we show that the termination condition given on Line~\ref{algline:neighborhood-termination} of~\Cref{alg:neighborhood} is safe. This will enable us to assume not only that $\ccg{G+ab}{b} < \ccg{G+ab}{a}$ (by~\Cref{lemma:no-switching}), but also that this inequality holds for the duration of the algorithm.

\begin{restatable}{lemma}{terminationlemma}\label{lemma:termination}
    Let $G = (V, E)$, $a, b \in V$, $ab \in E$, $\ccg{G}{b} < \ccg{G}{a}$, $\ccrg{G}{a}{b} < \frac{2}{3}$, and $u \in N_G(b)$. If $\ccg{G+au}{a} \leq \ccg{G+au}{b}$, then $\ccrg{G+au}{a}{b} \geq \frac{2}{3}$.
\end{restatable}

We remark upon the intuition and briefly sketch the main idea of the proof. Once again, we are motivated by the observation that if $a$ and $b$ are adjacent, then their distances to a given vertex may differ by at most one.
To leverage this, we will consider a partition of the vertices into three sets $A$, $B$, and $C$: those closer to $a$ than to $b$, those closer to $b$ than to $a$, and those equidistant from both $a$ and $b$, respectively.
The key step in the argument is to show that if $b$'s closeness centrality is \emph{much} less than that of $a$, i.e., less by a factor of $\frac{2}{3}$, then more than half of the vertices must be contained in $B$.
With sufficient attention to detail in the arithmetic, this step will follow by combining the triangle inequality with the definitions of $A$, $B$, and $C$.
Then, assuming toward a contradiction that the lemma is false, we repeat the argument to show that $A$ also contains more than half of the vertices, which is a contradiction since $A$ and $B$ are by definition disjoint.
The style of the arithmetic needed to formalize this argument is not dissimilar to that used in the proof of~\Cref{lemma:no-switching}; see~\Cref{appendix:termination-lemma} for the full proof.


\subsection{Analysis II: Identifying a Sufficient Condition}

What remains is to analyze the behavior of the loop in~\Cref{alg:neighborhood}.
We observe that if at any point $N_{G+S}(b) \setminus N_{G+S}(a) = \emptyset$, then $\ccg{G+S}{a} \leq \ccg{G+S}{b}$.
We may therefore assume that the number of edge additions $k < |N_{G}(b) \setminus N_G(a)|$. Otherwise the condition on Line~\ref{algline:neighborhood-termination} will be triggered and, by~\Cref{lemma:termination}, the returned solution
will be a $\frac{2}{3}$-approximation. This implies that the vertex $u$ on Line~\ref{algline:neighborhood-vertex-selection} is well-defined.

To complete the analysis, we need to show that the set $S$ returned on Line~\ref{algline:neighborhood-final-return} yields a $\frac{6}{11}$-approximation.
We begin by showing a sufficient condition: a large value of $k + |N_G(a)|$. 

\begin{restatable}{lemma}{lemmaarbitraryedges}\label{lemma:arbitrary-edges}
    Let $(G = (V, E), a, b, k)$ be an instance of \CRI{} with $\ccg{G}{b} \leq \ccg{G}{a}$. If $k + |N_G(a)| \geq \frac{n}{6}$, then~\Cref{alg:neighborhood} returns a $\frac{6}{11}$-approximation.
\end{restatable}

Let $S$ be the set of edges returned by~\Cref{alg:neighborhood}. Once again, we are leveraging a straightforward observation: for every neighbor $v$ of $a$ in $G+S$, $\distg{G+S}{a}{v} \leq \distg{G+S}{b}{v}$. Meanwhile, since the algorithm adds the edge $ab$ if it is not already present, for all other vertices $v$ we have $\distg{G+S}{a}{v} \leq \distg{G+S}{b}{v} + 1$. Thus, if $a$ has \emph{many}, i.e., $\frac{n}{6}$, neighbors in $G+S$, it is not possible for $\ccg{G+S}{b}$ to be much less than $\ccg{G+S}{a}$.
The proof proceeds by applying the above inequalities to obtain an upper bound for $\ccg{G+S}{a}$ in terms of $\ccg{G+S}{b}$; see~\Cref{appendix:sufficient-condition}.

\subsection{Analysis III: Proving the Sufficient Condition Occurs}

By~\Cref{lemma:termination}, if~\Cref{alg:neighborhood} returns on Line~\ref{algline:neighborhood-termination-line}, then it produces a~$\frac{2}{3}$-approximation. Similarly, by~\Cref{lemma:no-switching}, if~\Cref{alg:neighborhood} returns on Line~\ref{algline:neighborhood-no-switching-return}, then it produces a~$\frac{6}{11}$-approximation.
Thus, we may assume that~\Cref{alg:neighborhood} returns on Line~\ref{algline:neighborhood-final-return}.
By design, if $G$ is already $\frac{6}{11}$-approximate, then we are guaranteed to succeed.
Recall that other than the edge $ab$, every edge added by~\Cref{alg:neighborhood} is from $a$ to the neighborhood of $b$. Thus, after the edge $ab$ is added, the closeness centrality of $b$ does not change, while that of $a$ monotonically decreases.
Then we may conclude that if $G+ab$ is $\frac{6}{11}$-approximate, our algorithm returns a $\frac{6}{11}$-approximation.
It now suffices to show that if neither $G$ nor $G+ab$ is~$\frac{6}{11}$-approximate, then the sufficient condition given by~\Cref{lemma:arbitrary-edges} must occur.

\begin{restatable}{lemma}{mainlemma}\label{lemma:Bp-big}
    Let $(G = (V, E), a, b, k)$ be an instance of~\CRI{} with $\ccg{G}{b} < \ccg{G}{a}$ and $\ccg{G+ab}{b} < \ccg{G+ab}{a}$. Let $G^*$ be an optimal solution, and assume that $\max\{\ccrg{G}{a}{b}$,~$\ccrg{G+ab}{a}{b}\} < \frac{6}{11}\cdot\ccrg{G^*}{a}{b}$. 
    Then $k + |N_G(a)| \geq \frac{n}{6}$.
\end{restatable}

The full proof is in~\Cref{appendix:main-lemma}. Here, we sketch the argument. Recall the notation $B^p = N_G(b) \setminus N_G[a]$ for the private neighbors of $b$.
We begin by establishing some additional structure.

\begin{restatable}{claim}{setsizesclaim} \label{claim:set-sizes}
    Under the hypotheses of the lemma, all of the following hold:
    \begin{enumerate}[(i)]\itemsep=.7em
        \item $\ccg{G+ab}{b} < \frac{6(n-2)}{5}$, 
        \item $\ccrg{G^*}{a}{b} > \frac{11}{12}$, and
        \item $|B^p| > \frac{4(n-2)}{5}$.
    \end{enumerate}
\end{restatable}
To see part (i), observe that since $\distg{G+ab}{a}{b} = 1$, the triangle inequality yields $\ccg{G+ab}{a} \leq \ccg{G+ab}{b} + n - 2$. This bound implies that if $\ccg{G+ab}{b}$ is large, then the rato $\ccg{G+ab}{b}/\ccg{G+ab}{a}$ must be large, i.e., at least $\frac{6}{11}$, in contradiction of our hypothesis. The particular quantity $6(n-2)/5$ is a consequence of the arithmetic used to check this intuition.
Part (ii) follows from the fact that $\ccrg{G+ab}{a}{b} > \frac{1}{2}$ (recall~\Cref{sec:half-approx}), and the assumption that $G+ab$ is not $\frac{6}{11}$-approximate.
Part (iii) follows from an argument in the style of~\Cref{claim:switching-large-set}.

We will also need an upper bound on $\ccg{G^*}{b}$. 
This bound is a consequence of several creative combinations of~\Cref{claim:set-sizes} and the triangle inequality.

\begin{restatable}{claim}{claimbupperbound}\label{claim:b-upper-bound}
    $\ccg{G^*}{b} < d_{G^*}(a, b) + \frac{4(n - 2)}{5} + (d_{G^*}(a, b) + 1)\cdot\frac{n - 2}{5}$. 
\end{restatable}

\smallskip
Observe that condition (ii) of~\Cref{claim:set-sizes} implies the inequality $\ccg{G^*}{a} \leq \frac{12}{11}\cdot\ccg{G^*}{b}$. Our strategy will be to find lower bounds on $\ccg{G^*}{a}$.
We will substitute these bounds into the LHS of the preceding inequality. Meanwhile, we will substitute the bound obtained from~\Cref{claim:b-upper-bound} into the RHS.
Arithmetic manipulations will then allow us to derive our desired conclusions.\looseness=-1

It is easiest to perform this analysis with a specific value of $d_{G^*}(a, b)$ in mind, since then (using the triangle inequality) we have some control over the ratios achieved on individual vertices.
We begin by combining Claims~\ref{claim:set-sizes} and~\ref{claim:b-upper-bound} to show that $d_{G^*}(a, b) \in \{1, 2\}$.
Intuitively, the claim is true because, by the triangle inequality, for each neighbor $v$ of $b$, $\distg{G^*}{a}{v} \geq \distg{G^*}{a}{b} - 1$.
Since $b$ has many neighbors by condition (iii) of~\Cref{claim:set-sizes}, if $\distg{G^*}{a}{b}$ is large then it will not be possible for the optimum objective value to be close to one, i.e., for condition (ii) of~\Cref{claim:set-sizes} to hold.
\begin{restatable}{claim}{distancethreeclaim} \label{claim:distance-3}
    There is no optimal solution $G^*$ with $\distg{G^*}{a}{b} \geq 3$.
\end{restatable}

\smallskip
To formally prove~\Cref{claim:distance-3}, we first use condition (iii) of~\Cref{claim:set-sizes} to write a lower bound for $\ccg{G^*}{a}$ in terms of $n$ and $\distg{G^*}{a}{b}$. Observe that~\Cref{claim:b-upper-bound} gives us an upper bound for $\ccg{G^*}{b}$ in terms of these same quantities.
The result then follows by substituting these bounds into the inequality $\ccg{G^*}{a} \leq \frac{12}{11}\cdot\ccg{G^*}{b}$ (condition (ii) of~\Cref{claim:set-sizes}) and solving for $\distg{G^*}{a}{b}$.

Thus, either $\distg{G^*}{a}{b} = 1$ or $\distg{G^*}{a}{b} = 2$. We handle the easy case first. 
\begin{restatable}{claim}{distanceoneclaim} \label{claim:distance-1}
    If there exists an optimal $G^*$ with $\distg{G^*}{a}{b} = 1$, then $|N_G(a)| + k > \frac{2n}{5} > \frac{n}{6}$.
\end{restatable}

\smallskip
The proof of~\Cref{claim:distance-1} is straightforward. Assume the claim is false. Then we may write an obvious lower bound for $\ccg{G^*}{a}$ by noting (trivially) that $a$ has distance $1$ to all of its neighbors and distance at least $2$ to all of its non-neighbors.
Meanwhile, we substitute $\distg{G^*}{a}{b} = 1$ into~\Cref{claim:b-upper-bound} to obtain an upper bound on $\ccg{G^*}{b}$. Using condition (ii) of~\Cref{claim:set-sizes}, we substitute these bounds into the inequality $\ccg{G^*}{a} \leq \frac{12}{11}\cdot\ccg{G^*}{b}$. Simplifications of the resulting inequality yield an absurdity.

In some sense, Claims~\ref{claim:distance-3} and~\ref{claim:distance-1} are ``opposite ends'' of an analysis resulting from conditions (ii) and (iii) of~\Cref{claim:set-sizes}.
If $\distg{G^*}{a}{b}$ is very large (\Cref{claim:distance-3}), then we can use the fact that $b$ has many neighbors (together with the triangle inequality) to write a prohibitively large lower bound on $\ccg{G^*}{a}$.
On the other hand, if $\distg{G^*}{a}{b} = 1$, then the upper bound on $\ccg{G^*}{b}$ yielded by \Cref{claim:b-upper-bound} is so small that $a$ cannot possibly ``catch up'' unless it is allowed to have many neighbors in $G^*$. 

The remaining case, when $\distg{G^*}{a}{b} = 2$, requires the ``middle ground'' between these two approaches. However, the overall strategy is unchanged. We substitute $\distg{G^*}{a}{b} = 2$ into~\Cref{claim:b-upper-bound} to obtain an upper bound on $\ccg{G^*}{b}$. Meanwhile, we use condition (iii) of~\Cref{claim:set-sizes} to obtain a lower bound on $\ccg{G^*}{a}$. By substituting these inequalities into condition (ii) of~\Cref{claim:set-sizes}, we derive the desired lower bound on $k$, completing the proof of~\Cref{lemma:Bp-big}.
See~\Cref{appendix:main-lemma} for the details.

Combining Lemmas~\ref{lemma:no-switching},~\ref{lemma:termination},~\ref{lemma:arbitrary-edges}, and~\ref{lemma:Bp-big}, we achieve the main result. 
\begin{theorem}\label{thm:approx}
    \Cref{alg:neighborhood} is a $\frac{6}{11}$-approximation for \CRI{}, running in $O((n+m)\log k)$ time.
\end{theorem}

\section{Conclusion and Open Problems}\label{sec:conclusion}

In this paper, we used closeness centrality to represent social capital in networks and introduced several novel problems aimed at equalizing the capital of two specified nodes via edge additions. For \CRI{}, we proved \cclass{NP}-hardness for all target ratios $\tau \in (\frac{1}{2}, 1]$, while observing that the problem is trivially solvable $\tau \leq \frac{1}{2}$: just add the edge $ab$ between the relevant vertices. We also showed how to quantify and harness the structure of non-trivial instances, i.e., instances where adding $ab$ is not sufficient, leading to a scalable (quasilinear time) $\frac{6}{11}$-approximation. Additionally, we proved the problem is inapproximable within any factor greater than $\frac{5e}{5e + 1 - \varepsilon} \approx .932$. For \CGM{}, we also showed \cclass{NP}-hardness and further ruled out any multiplicative approximation.
Finally, we proposed an extension,~\APCRI{}, in which the goal is to equalize the closeness centralities of all node pairs, rather than two specified nodes. We established the \cclass{NP}-hardness of this problem. 

Amongst the advantages of our framework are the elegance of our combinatorial model, and the technical simplicity of our algorithm and its accompanying analysis. The choice of closeness centrality as the measure for social capital was motivated in part by these considerations, but this does not mean that it is the only compelling choice. Indeed, we suggest that while much algorithmic work has been devoted to graph modification with the intent of optimizing some characteristic of a given node, modification to \emph{equalize} node characteristics is both well-motivated and relatively under-explored, whatever the characteristic of interest. For example, it may be interesting to consider variants of our framework that equalize betweenness centrality~\cite{bavelas1950communication} or PageRank~\cite{page1999pagerank}.

We also emphasize that our model, in which a pair of nodes is assumed as input, is intended as a starting point for a broader set of algorithmic questions motivated by the goal of promoting fairness across a social network.
We have already taken one step in this direction by studing~\APCRI{}. A more direct generalization of~\CRI{} considers two \emph{groups} of vertices.
The goal is to add at most $k$ edges to the graph to maximize the ratio between the least central node in each group. 

\optproblemdef{\textsc{Group Closeness Ratio Improvement}}{A graph $G=(V,E)$, two groups $A,B\subseteq V$, $A \cap B = \emptyset$, and $k \in \mathbb N$.}{Find $S \subseteq V^2 \setminus E$, $|S| \leq k$, maximizing $\ccrg{G+S}{\ccworst{G+S}{A}}{\ccworst{G+S}{B}}$.}

This is a generalization of \CRI{}, but it seems that new insights will be needed to tackle it. For example, if $|A|, |B|$ are unbounded with respect to $k$, then there is not a clear way to lift the trivial approach outlined in~\Cref{sec:half-approx}, which is also a crucial component of the algorithm in~\Cref{sec:approximation-alg}. Even when $|A|, |B|$ are bounded, it is not clear that useful generalizations of our Lemmas~\ref{lemma:no-switching} and~\ref{lemma:termination} can be obtained, and the lower bound on $k$ required for a straightforward generalization of~\Cref{lemma:arbitrary-edges} is likely to be very large.
One might also consider optimizing the ratio of mean or best centralities within each group; the minimax formulation presented above is common in the fairness literature, e.g.,~\cite{bashardoust}, usually with the justification that it focuses on aiding the most disadvantaged individuals.

\section*{Acknowledgements}
This work was supported in part by the National Science Foundation under awards IIS-1955321 to Sorelle A. Friedler and IIS-1956286 to Blair D. Sullivan.

\appendix
\section{Omitted Details from~\Cref{sec:approximation-alg}}\label{appendix:early-termination-lemmas}

\subsection{Proof of~\Cref{lemma:termination}}\label{appendix:termination-lemma}

\terminationlemma*
\begin{proof}
    Consider a partition of the vertices $V$ into five sets: $\{a\}$, $\{b\}$, other vertices
    closer to $a$ than to $b$, closer to $b$ than to $a$, and
    those equidistant from $a$ and $b$. Formally, the latter three sets are
    $A_G \vcentcolon = \{v \in V\setminus\{a, b\} \colon \distg{G}{a}{v} < \distg{G}{b}{v}\}$,
    $B_G \vcentcolon = \{v \in V\setminus\{a, b\} \colon \distg{G}{a}{v} > \distg{G}{b}{v}\}$,  
    and
    $C_G \vcentcolon = \{v \in V\setminus\{a, b\} \colon \distg{G}{a}{v} = \distg{G}{b}{v}\}$.  

    Let $A_{G+au}$, $B_{G+au}$, and $C_{G+au}$ be defined similarly.
    We will now use these definitions and the triangle inequality to show that $B_G$ is large.
    
    \begin{restatable}{claim}{claimterminationbigset}\label{claim:termination-big-set}
        $|B_G| > \frac{n - 2}{2}$.
    \end{restatable}
    \begin{proof}[Proof of~\Cref{claim:termination-big-set}]
        We first use the definitions of $A_G$ and $C_G$ to write an expression for $\ccg{G}{b}$.
        \begin{align*}
            \ccg{G}{b} &= \distg{G}{b}{b} + \distg{G}{b}{a} + \meandistg{G}{b}{B_G}\cdot|B_G| + \meandistg{G}{b}{A_G \cup C_G}\cdot|A_G \cup C_G| 
        \end{align*}
        
        Next, we will give an upper bound for $\ccg{G}{a}$. Using the triangle inequality and the fact that $ab \in E(G)$ (by hypothesis), we observe that $\meandistg{G}{a}{B_G} \leq \meandistg{G}{b}{B_G} + 1$ (in fact equality holds, using the definition of $B_G$).
        Then we have
        \begin{align*}
            \ccg{G}{a} &= \distg{G}{a}{a} + \distg{G}{a}{b} + \meandistg{G}{a}{B_G}\cdot|B_G| + \meandistg{G}{a}{A_G \cup C_G}\cdot|A_G \cup C_G| \\
            &\leq \distg{G}{a}{b} + (\meandistg{G}{b}{B_G} + 1)\cdot|B_G| + \meandistg{G}{a}{A_G \cup C_G}\cdot|A_G \cup C_G|. 
        \end{align*}
    
        Now recall that, by hypothesis, $\ccg{G}{b} < \frac{2}{3}\cdot\ccg{G}{a}$. Then by substituting our bounds we have
        \begin{align*}
            \distg{G}{a}{b} + \meandistg{G}{b}{B_G}&\cdot|B_G| + \meandistg{G}{a}{A_G \cup C_G}\cdot|A_G \cup C_G| \\ &< \frac{2}{3}\cdot\bigg[\distg{G}{a}{b} + (\meandistg{G}{b}{B_G} + 1)\cdot|B_G| + \meandistg{G}{a}{A_G \cup C_G}\cdot|A_G \cup C_G| \bigg] \\
            &= \frac{2}{3}\cdot\distg{G}{a}{b} + \frac{2}{3}\cdot\meandistg{G}{b}{B_G}\cdot|B_G| + \frac{2}{3}\cdot|B_G| + \frac{2}{3}\cdot \meandistg{G}{a}{A_G \cup C_G}\cdot|A_G \cup C_G|.
        \end{align*}
        Combining terms, we have
        \begin{align*}
            \frac{1}{3}\cdot\distg{G}{a}{b} + \frac{1}{3}\cdot\meandistg{G}{a}{A_G \cup C_G}\cdot|A_G \cup C_G| &< \bigg[\frac{2}{3} - \frac{1}{3}\cdot\meandistg{G}{b}{B_G}\bigg]\cdot|B_G|,
        \end{align*}
        or, by dropping the first term on the LHS,
        \begin{align*}
            \frac{1}{3}\cdot\meandistg{G}{a}{A_G \cup C_G}\cdot|A_G \cup C_G| &< \bigg[\frac{2}{3} - \frac{1}{3}\cdot\meandistg{G}{b}{B_G}\bigg]\cdot|B_G|.
        \end{align*}
        Now, we complete the proof by assuming toward a contradiction that the claim is false, i.e., that $|B_G| \leq \frac{n-2}{2}$, and consequently $\frac{n-2}{2} \leq |A_G \cup C_G|$.
        Then we substitute these inequalities as follows:
        \begin{align*}
            \frac{1}{3}\cdot\meandistg{G}{a}{A_G \cup C_G}\cdot\frac{n-2}{2} &< \bigg[\frac{2}{3} - \frac{1}{3}\cdot\meandistg{G}{b}{B_G}\bigg]\cdot\frac{n-2}{2},
        \end{align*}
        which implies
        \begin{align*}
            \meandistg{G}{a}{A_G \cup C_G} + \meandistg{G}{b}{B_G} &< 2.            
        \end{align*}
        This is a contradiction, since $a \notin A_G \cup C_G$ and $b \notin B_G$, implying that both summands on the LHS are at least $1$. This completes the proof of the claim.
    \end{proof}
        
    Now we complete the proof of the lemma by assuming toward a contradiction that $\ccg{G+au}{a} < \frac{2}{3}\cdot\ccg{G+au}{b}$.
    Then, similarly to~\Cref{claim:termination-big-set}, $|A_{G+au}| > \frac{n - 2}{2}$.
    To see why this is a contradiction, observe that since $u \in N_G(b)$, $A_{G+au} \subseteq A_G$ (in fact the sets are equal).
    Then we have $|A_{G}|, |B_{G}| > \frac{n - 2}{2}$, but these sets are disjoint by definition, and their union has cardinality at most $n - 2$, so we have a contradiction. 
\end{proof}

\subsection{Proof of~\Cref{lemma:arbitrary-edges}}\label{appendix:sufficient-condition}

\lemmaarbitraryedges*
\begin{proof}
    Let $S$ be the set of edge additions returned by~\Cref{alg:neighborhood}. By~\Cref{lemma:termination}, we may assume that $\ccg{G+S}{b} < \ccg{G+S}{a}$.
    Thus, it is sufficient to show that $\ccg{G+S}{a} \leq \frac{11}{6}\ccg{G+S}{b}$.

    Let $A = N_{G+S}(a) \setminus \{b\}$. Note that by~\Cref{lemma:no-switching,lemma:termination}, we may assume that $|S| = k$, i.e., that the algorithm did not terminate early.
    It follows that $|A| \geq \frac{n}{6} - 1$. 
    Setting $B = V \setminus (A \cup \{a, b\})$, we have $|B| \leq \frac{5n}{6} - 1$.
    
    We now wish to find an upper bound for $\ccg{G+S}{a}$ in terms of $n$ and $\ccg{G+S}{b}$.
    We begin with the following:
    \begin{align*}
        \ccg{G+S}{a} &= d_{G+S}(a, a) + d_{G+S}(a, b) + \sum_{v \in A}d_{G+S}(a, v) + \sum_{v \in B} d_{G+S}(a, v) \\
        &= d_{G+S}(b, b) + d_{G+S}(b, a) + \sum_{v \in A}d_{G+S}(a, v) + \sum_{v \in B} d_{G+S}(a, v). \\
    \end{align*}
    By~\Cref{lemma:no-switching}, we may assume that $d_{G+S}(a, b) = 1$.
    Then for each $v \in B$, $d_{G+S}(a, v) \leq d_{G+S}(b, v) + 1$. 
    We use this to make a substitution in the fourth summand of the RHS:
    \begin{align*}
        \ccg{G+S}{a} &\leq d_{G+S}(b, b) + d_{G+S}(b, a) + \sum_{v \in A}d_{G+S}(a, v) + \sum_{v \in B} (d_{G+S}(b, v) + 1).
    \end{align*}

    Meanwhile, since $b \notin A$, we have $d_{G+S}(a, v) \leq d_{G+S}(b, v)$ for all $v \in A$.
    Then we may perform a substitution in the third summand to obtain the bound we desire.
    \begin{align*}
        \ccg{G+S}{a} &\leq d_{G+S}(b, b) + d_{G+S}(b, a) + \sum_{v \in A}d_{G+S}(b, v) + \sum_{v \in B} (d_{G+S}(b, v) + 1) \\
        &= \ccg{G+S}{b} + |B| \\
        &\leq \ccg{G+S}{b} + \frac{5n}{6} - 1. \\
    \end{align*}

    We now complete the proof in two cases. In the first case, we assume that $\ccg{G+S}{b} \geq n$.
    Then we may write $\ccg{G+S}{b} = n + x$ for some $x \geq 0$. We now have
    \begin{align*}
        \ccrg{G+S}{a}{b} = \frac{\ccg{G+S}{b}}{\ccg{G+S}{a}} \geq \frac{\ccg{G+S}{b}}{\ccg{G+S}{b} + \frac{5n}{6} - 1} &\geq \frac{\ccg{G+S}{b}}{\ccg{G+S}{b} + \frac{5n}{6}} \\
        &= \frac{n + x}{n + x + \frac{5n}{6}} \\
        &\geq \frac{n}{n + \frac{5n}{6}} \quad\quad(\text{by~\Cref{obs:adding-to-fractions}})\\
        &= \frac{6}{11}.
    \end{align*}

    In the other case, we have $\ccg{G+S}{b} = n - 1$. Then $b$ is adjacent to every other vertex in the graph, meaning that $a$ has distance at most $2$ to every vertex in $B$. Then
    \begin{align*}
        \ccg{G+S}{a} &\leq d_{G+S}(a,a) + 1\cdot\left((n - 1) - |B|\right) + 2\cdot|B| \\
        &= n + |B| - 1 \leq n + \frac{5n}{6} - 2 = \frac{11n}{6} - 2 < \frac{11n}{6} - \frac{11}{6} = \frac{11}{6}\cdot\ccg{G+S}{b}.        
    \end{align*}
    This completes the proof.
    \end{proof}

\subsection{Full Proof of~\Cref{lemma:Bp-big}}\label{appendix:main-lemma}

\mainlemma*
\begin{proof}
    Recall the notation $B^p = N_G(b) \setminus N_G[a]$ for the private neighbors of $b$.
    We begin by establishing some additional structure.

    \setsizesclaim*
    \begin{proof}[Proof of~\Cref{claim:set-sizes}]
        Part (i) is a consequence of the observation that $\ccg{G+ab}{a} \leq \ccg{G+ab}{b} + n - 2$, which follows from the triangle inequality and the fact that $\distg{G+ab}{a}{b} = 1$. Thus, if $\ccg{G+ab}{b} \geq \frac{6(n-2)}{5}$, then we may write $\ccg{G+ab}{b} = \frac{6(n-2)}{5} + x$ for some non-negative $x$. Then
        \[ 
            \ccrg{G+ab}{a}{b} = \frac{\ccg{G+ab}{b}}{\ccg{G+ab}{a}} \geq \frac{\frac{6(n-2)}{5} + x}{\frac{11(n-2)}{5} + x} \geq \frac{6}{11}, 
        \]
        where in the final inequality we have used~\Cref{obs:adding-to-fractions}. This contradicts the hypothesis that $\ccrg{G+ab}{a}{b} < \frac{6}{11}\ccrg{G^*}{a}{b}$, since $\ccrg{G^*}{a}{b} \leq 1$.
    
        Part (ii) follows from the observation that $\ccrg{G+ab}{a}{b} > \frac{1}{2}$ (recall~\Cref{sec:half-approx}). Then since $\ccrg{G+ab}{a}{b} < \frac{6}{11}\ccrg{G^*}{a}{b}$, we have $\ccrg{G^*}{a}{b} > \frac{11}{6}\cdot\frac{1}{2} = \frac{11}{12}$.
    
        To prove part (iii), it is sufficient to show that $b$ has strictly more than $\frac{4(n-2)}{5}$ private neighbors in $G+ab$, since any private neighbor of $b$ in $G+ab$ is also a private neighbor of $b$ in $G$.
        The remainder of the argument is identical to the proof of~\Cref{claim:switching-large-set}, except that the roles of $a$ and $b$ have been reversed.
    \end{proof}

    \bigskip
    Our strategy will now be to use~\Cref{claim:set-sizes} to find upper bounds on $\ccg{G^*}{b}$ and lower bounds on $\ccg{G^*}{a}$, where the latter will be expressed as a function of $k$.
    We will then leverage the gap between these bounds, together with condition (ii) of~\Cref{claim:set-sizes}, to prove that $k + |N_G(a)|$ must be large.
    It is easiest to perform this analysis with a specific value of $d_{G^*}(a, b)$ in mind, since then (using the triangle inequality) we have some control over the ratios achieved on individual vertices.
    We begin by obtaining an upper bound on $\ccg{G^*}{b}$.

    \claimbupperbound*
    \begin{proof}[Proof of~\Cref{claim:b-upper-bound}]        
        Note that if $\distg{G^*}{a}{b} = 1$, then the claim follows from part (i) of~\Cref{claim:set-sizes}, using monotonicity to observe $\ccg{G^*}{b} \leq \ccg{G+ab}{b}$.
        Thus, we may assume that $\distg{G^*}{a}{b} \geq 2$.
        Let $O = V \setminus (B^p \cup \{a, b\})$. We claim it is sufficient to show that
        for each $v \in O$,
        \begin{align}\label{eq:distance-blowup}
            d_{G^*}(b, v) \leq d_{G+ab}(b, v) + d_{G^*}(a, b) - 1.
        \end{align}
        To see this, begin with a straightforward expression:
          $\ccg{G^*}{b} = \distg{G^*}{a}{b} + |B^p| + \sum_{v \in O} \distg{G^*}{b}{v}$.  
        Then using Inequality~(\ref{eq:distance-blowup}), we have
        \[
          \ccg{G^*}{b} \leq \distg{G^*}{a}{b} + |B^p| + \sum_{v \in O} \distg{G+ab}{b}{v} + (\distg{G^*}{a}{b} - 1)\cdot|O|.
        \] 
        By parts (i) and (iii) of~\Cref{claim:set-sizes}, $\sum_{v \in O}d_{G+ab}(b, v) < \frac{2(n-2)}{5}$, giving
        \[
          \ccg{G^*}{b} < \distg{G^*}{a}{b} + |B^p| + \frac{2(n-2)}{5} + (\distg{G^*}{a}{b} - 1)\cdot|O|.
        \]
        Now, using part (iii) of~\Cref{claim:set-sizes} we write $|B^p| = \frac{4(n-2)}{5} + x$ and $|O| = \frac{n-2}{5} - x$ for some positive $x$.
        Substituting, we have
        \begin{align*}
            \ccg{G^*}{b} &< \distg{G^*}{a}{b} + \frac{4(n-2)}{5} + x + \frac{2(n-2)}{5} + (\distg{G^*}{a}{b} - 1)\cdot\left(\frac{n-2}{5} - x\right) \\
            &= \distg{G^*}{a}{b} + \frac{4(n-2)}{5} + (\distg{G^*}{a}{b} + 1)\cdot\frac{n-2}{5} + (2 - \distg{G^*}{a}{b})x \\
            &\leq \distg{G^*}{a}{b} + \frac{4(n-2)}{5} + (\distg{G^*}{a}{b} + 1)\cdot\frac{n-2}{5} \quad\quad(\text{using $\distg{G^*}{a}{b} \geq 2$}).
        \end{align*}
        This proves that it suffices to show Inequality~(\ref{eq:distance-blowup}).
        Consider $v \in O$. Observe that if $d_{G}(b, v) \leq d_{G}(a, v)$, then $d_{G}(b, v) = d_{G+ab}(b, v)$.
        Then we simply observe that $d_{G^*}(b, v) \leq d_{G}(b, v) = d_{G+ab}(b,v) \leq d_{G+ab}(b,v) + d_{G^*}(a, b) - 1$, as desired.
        Thus, we may assume that $d_G(b, v) > d_G(a, v)$, in which case $d_{G+ab}(b, v) = 1 + d_{G+ab}(a, v)$.
        Then by the triangle inequality,
        \begin{align*}
            d_{G^*}(b, v) \leq d_{G^*}(a, b) + d_{G^*}(a, v) \leq d_{G^*}(a, b) + d_{G}(a, v) &= d_{G^*}(a, b) + d_{G+ab}(a, v) \\
            &= d_{G^*}(a, b) + d_{G + ab}(b, v) - 1,
        \end{align*}
        which completes the proof.
    \end{proof}

    We now combine Claims~\ref{claim:set-sizes} and~\ref{claim:b-upper-bound} to show that $d_{G^*}(a, b) \in \{1, 2\}$.
    \distancethreeclaim*
    \begin{proof}[Proof of~\Cref{claim:distance-3}]
        By the triangle inequality, for each $v \in B^p$, $\distg{G}{a}{v} \geq \distg{G^*}{a}{b} - 1$. Moreover, for every vertex $v \in O = V \setminus (B^p \cup \{a,b\})$, trivially $\distg{G^*}{a}{v} \geq 1$.
        Using part (iii) of~\Cref{claim:set-sizes}, we write $|B^p| = \frac{4(n-2)}{5} + x$ and $|O| = \frac{n-2}{5} - x$, for some positive $x$. Then
        \begin{align*}
            \ccg{G^*}{a} &\geq \distg{G^*}{a}{b} + (d_{G^*}(a, b) - 1)\cdot\left(\frac{4(n-2)}{5} + x\right) + \frac{n - 2}{5} - x \\
            &= \distg{G^*}{a}{b} + (d_{G^*}(a, b) - 1)\cdot\frac{4(n-2)}{5} + \frac{n - 2}{5} + (\distg{G^*}{a}{b} - 2)x \\
            &> \distg{G^*}{a}{b} + (d_{G^*}(a, b) - 1)\cdot\frac{4(n-2)}{5} + \frac{n - 2}{5} \quad\quad(\text{If $\distg{G^*}{a}{b} \leq 2$ we're done.}) \\
            &= \distg{G^*}{a}{b} + (4d_{G^*}(a, b) - 3)\cdot\frac{n - 2}{5}.
        \end{align*}
        Meanwhile, by~\Cref{claim:b-upper-bound},
        \[
          \ccg{G^*}{b} < d_{G^*}(a, b) + (d_{G^*}(a, b) + 5)\cdot\frac{n-2}{5}.
        \]
        Finally, by part (ii) of~\Cref{claim:set-sizes},
        \begin{align*}
            \ccg{G^*}{a} &\leq \frac{12}{11}\cdot\ccg{G^*}{b}.
        \end{align*}
        Now we complete the proof by substituting our bounds and solving for $d_{G^*}(a, b)$.
        \begin{align*}            
            d_{G^*}(a, b) + (4d_{G^*}(a, b) - 3)\cdot\frac{n - 2}{5} &< \frac{12}{11}\bigg[d_{G^*}(a, b) + (d_{G^*}(a, b) + 5)\cdot\frac{n-2}{5}\bigg] \\
            \frac{32}{11}\cdot d_{G^*}(a, b)\cdot \frac{n-2}{5} &< \frac{1}{11}\cdot d_{G^*}(a, b) + \frac{93}{11}\cdot \frac{n-2}{5} \\
            32\cdot d_{G^*}(a, b)\cdot\frac{n-2}{5} &< d_{G^*}(a, b) + 93\cdot \frac{n-2}{5} \\
            d_{G^*}(a, b) &< \frac{93\cdot\frac{n-2}{5}}{32\cdot\frac{n-2}{5} - 1} \\
            d_{G^*}(a, b) &< \frac{93\cdot\frac{n-2}{5}}{32\cdot\frac{n-2}{5} - 1\cdot\frac{n-2}{5}} \quad\quad (\text{whenever } n \geq 7) \\
            d_{G^*}(a, b) &< \frac{93}{31} = 3.
        \end{align*}
        To complete the proof, note that when $n < 7$ we may assume that $d_{G^*}(a, b) < 3$. Otherwise, it is impossible to satisfy part (iii) of~\Cref{claim:set-sizes}.
    \end{proof}

    Thus, either $\distg{G^*}{a}{b} = 1$ or $\distg{G^*}{a}{b} = 2$. We handle the easy case first. 
    \distanceoneclaim*
    \begin{proof}[Proof of~\Cref{claim:distance-1}]
        Assume otherwise, i.e., that $|N_G(a)| + k \leq \frac{2n}{5}$. Then 
        $\ccg{G^*}{a} \geq \frac{2n}{5} + 2\left(\frac{3n}{5} - 1\right) = \frac{8n}{5} - 2$.  
        Meanwhile, since $d_{G^*}(a, b) = 1$, $\ccg{G^*}{b} \leq \ccg{G+ab}{b} < \frac{6(n - 2)}{5}$, where the latter inequality is from part (i) of~\Cref{claim:set-sizes}.
        But then by part (ii) of~\Cref{claim:set-sizes} we can obtain the following contradiction:
        \begin{align*}
            \ccg{G^*}{a} &< \frac{12}{11}\ccg{G^*}{b} \\
            \frac{8n}{5} - 2 &< \frac{12}{11}\left(\frac{6n}{5} - \frac{12}{5}\right) \\
            \frac{88n}{55} - \frac{110}{55} &< \frac{72n}{55} - \frac{124}{55} \\
            16n + 14 &< 0 \\
            n &< 0.
        \end{align*}
    \end{proof}

    It follows that in the remainder of the proof, we may assume that $\distg{G^*}{a}{b} = 2$. Our goal is to show that if $\ccrg{G+ab}{a}{b} < \frac{6}{11}\cdot\ccrg{G^*}{a}{b}$, then $k + |N_G(a)| \geq \frac{n}{6}$.
    We begin by using~\Cref{claim:b-upper-bound} to establish an upper bound for $\ccg{G^*}{b}$:
      $\ccg{G^*}{b} < \frac{7(n-2)}{5} + 2$.  
    
    Now we will find a lower bound for $\ccg{G^*}{a}$. Once again we define $O = V \setminus (B^p \cup \{a,b\})$, and use part (iii) of~\Cref{claim:set-sizes} to write $|B^p| = \frac{4(n-2)}{5} + x$ and $|O| = \frac{n-2}{5} - x$ for some positive $x$.
    We will also use the facts that at least $|B^p| - k$ private neighbors of $b$ in $G$ remain private neighbors of $b$ in $G^*$, and that $\meandistg{G^*}{a}{O} \geq 1$.
    \begin{align*}
        \ccg{G^*}{a} &\geq \distg{G^*}{a}{b} + 2\cdot\left(|B^p| - k\right) + k + \meandistg{G^*}{a}{O}\cdot|O| \\
        &\geq \distg{G^*}{a}{b} + 2\cdot|B^p| + 1\cdot|O| - k \\
        &= \distg{G^*}{a}{b} + \frac{8(n-2)}{5} + 2x + \frac{n-2}{5} - x - k \\
        &= 2 + \frac{9(n-2)}{5} - k + x
        > 2 + \frac{9(n-2)}{5} - k.
    \end{align*}

    We now complete the proof using part (ii) of~\Cref{claim:set-sizes}, which implies $\ccg{G^*}{a} < \frac{12}{11}\cdot\ccg{G^*}{b}$.
By substituting our upper bound on $\ccg{G^*}{b}$ and our lower bound on $\ccg{G^*}{a}$, we have
\begin{align*}
    2 + \frac{9(n-2)}{5} - k &< \frac{12}{11}\left(\frac{7(n-2)}{5} + 2\right),
\end{align*}
which simplifies to
\begin{align*}
    2 + \frac{99(n-2)}{55} - k &< \frac{84(n-2)}{55} + \frac{24}{11} \\
    \frac{15(n-2)}{55} - \frac{2}{11} &< k \\
    \frac{3(n - 2)}{11} - \frac{2}{11} &< k \\
    \frac{3n}{11} - \frac{8}{11} &< k \\
    \frac{n}{6} &< k \quad(\text{whenever }n > 3). 
\end{align*}
We note that the desired condition is trivially true when $n \leq 3$, since we may assume that $k \geq 1$. Thus, the proof is complete.
\end{proof}

\section{Deferred Proofs from~\Cref{sec:hardness-section}}\label{appendix:hardness}

\subsection{Proof of~\Cref{t1-hard}}\label{appendix:hardness-star}

All that remains is to prove~\Cref{claim:star-on-b}.

\setvertexincidentclaim*
\begin{proof}[Proof of~\Cref{claim:star-on-b}]
    We begin by creating a set $T'$ as follows. For each edge in $T^*$ that is incident on $b$, include it in $T'$. For each edge $xy\in T^*$ that is not incident on $b$, suppose without loss of generality that $\distg{G+T^*\setminus\{xy\}}{b}{x}\leq\distg{G+T^*\setminus\{xy\}}{b}{y}$, and then include the edge $by$ in $T'$. Now $|T'|\leq k$ and all edges in $T'$ are incident on $b$.

    Let $v \in V$, and let $P$ be a shortest path from $b$ to $v$ in $G + T^*$. If $P$ uses no edges from $T^*$, then $\distg{G+T'}{b}{v} \leq \distg{G+T^*}{b}{v}$. Otherwise, let $xy$ be the last edge from $T^*$ appearing along $P$ ($x$ may be $b$), so $P$ has length $\distg{G+T^*\setminus\{xy\}}{b}{x} + 1 + \distg{G}{y}{v}$.
    Then $\distg{G+T^*\setminus\{xy\}}{b}{y} > \distg{G+T^*\setminus\{xy\}}{b}{x}$, as otherwise there exists a path from $b$ to $v$ in $G+T^*$ with length at most $\distg{G+T^*\setminus\{xy\}}{b}{y} + \distg{G}{y}{v}$ which is strictly less than the length of $P$.
    Then by construction, the edge $by$ is in $T'$, so $\distg{G+T'}{b}{v} \leq \distg{G+T^*}{b}{v}$.

    We now use a second exchange argument on $T'$ to construct a set of at most $k$ edges $T$ such that every edge in $T$ is incident on $b$, all other endpoints of edges in $T$ are set vertices, and $\ccg{G+T}{b}\leq\ccg{G+T'}{b}\leq\ccg{G+T^*}{b}$.
    Let $bz_i$ be an edge (if one exists) in $T'$ which is incident on a $z_i$ in the independent set. Then $\ccg{G+T'}{b} = \ccg{G+T'\setminus\{bz_i\}}{b} - 1$. We may then replace $bz_i$ with $bs_j$, where $s_j$ is any set vertex not already an endpoint of an edge in $T'$, and achieve the same (or better) closeness centrality. Hence, we may assume that
    all edges in $T'$ are incident on $b$ and all other endpoints are either set or element vertices.

    Suppose that there are $k' \leq k$ element vertices incident on edges in $T'$. If $k' = 0$ then by setting $T = T'$ we are done.
    We induct on $k'$, assuming that any $T''$ with $k' - 1$ edges incident on element vertices may be transformed into the desired set $T$.
    Let $v_{i_1}$ be an element vertex which is an endpoint of an edge in $T'$, and let $C'$ be the set vertices which are endpoints of edges in $T'$.
    Since by hypothesis there is no set cover of size $k$, there exists an element $u_{i_2}$ which is not covered by the sets corresponding to $C'$ (it is possible that $i_1 = i_2$).
    Let $S_j$ be any set covering $u_{i_2}$, and set $T'' = (T' \setminus \{bv_{i_1}\}) \cup \{bs_j\}$.
    Then $\distg{G+T''}{b}{s_j} = \distg{G+T'}{b}{s_j} - 1$. If $i_1 = i_2$, then $\distg{G+T''}{b}{v_{i_1}} = \distg{G+T'}{b}{v_{i_1}} + 1$ and $\distg{G+T''}{b}{v} \leq \distg{G+T'}{b}{v}$ for all remaining vertices $v \notin\{v_{i_1}, s_j\}$.
    Otherwise, $\distg{G+T''}{b}{v_{i_1}} \leq \distg{G+T'}{b}{v_{i_1}} + 2$, $\distg{G+T''}{b}{v_{i_2}} = \distg{G+T'}{b}{v_{i_2}} - 1$, and $\distg{G+T''}{b}{v} \leq \distg{G+T'}{b}{v}$ for all remaining vertices $v \notin\{v_{i_1},v_{i_2}, s_j\}$.
    In either case, summing shows that $\ccg{G+T''}{b} \leq \ccg{G+T'}{b}$. The result now follows from the inductive hypothesis.
\end{proof}

\subsection{Proof of~\Cref{cor:gap-minimization-hard}}\label{appendix:hardness-gap-minimization}
\gapminimizationhard*
\begin{proof}
    Using the same construction as the proof of Theorem~\ref{t1-hard} and Figure~\ref{cri-instance}, we claim that there exists a set cover of $U$ of size at most $k$ if and only if there is a set $T$ of at most $k$ edges such that $|\ccg{G+T}{a}-\ccg{G+T}{b}|=0$. \\\\
    Suppose there is a set cover $C$ of $U$ of size at most $k$, then let $T = \{bs_j \colon S_j \in C\}$. Now, $\ccg{G+T}{a}=(n+k)+2+2m+3n=4n+2m+k+2$ and $\ccg{G+T}{b}=2(n+k)+2+2m-k+2n=4n+2m+k+2$, meaning $|\ccg{G+T}{a}-\ccg{G+T}{b}|=0$. Suppose instead that there is no collection of $k$ sets which covers $U$. We claim that for every set $T$ of size at most $k$,  $\ccg{G+T}{b} > \ccg{G}{a}\geq \ccg{G+T}{a}$. Let $T^*$ be a set of (at most) $k$ edge additions such that for all sets $T \subseteq V^2 \setminus E$ of size at most $k$, $\ccg{G+T^*}{b} \leq \ccg{G+T}{b}$. Then by the same argument used in the proof of~\Cref{t1-hard}, $\ccg{G+T^*}{b}\geq 2(n+k)+2+2m-k+2(n-1)+3(1)=4n+2m+k+3$, while $\ccg{G+T^*}{a}\leq \ccg{G}{a} = 4n+2m+k+2$. Consequently, $|\ccg{G+T}{a}-\ccg{G+T}{b}|\geq 1$ for all sets $T \subseteq V^2 \setminus E$ of size at most $k$. \\\\
    Since \CGM{} is \cclass{NP}-hard at parameter value 0, it admits no polynomial-time multiplicative approximation algorithm. Suppose that such an algorithm $\mathcal{A}$ did exist, achieving a gap at most $\alpha\cdot\opt(y)$ for some $\alpha$ possibly a function of the input $y$. Then consider an instance of \SC{} $x$ reduced (as we've described) to an instance $y$ of~\CGM{}. We've shown that if $x$ has a set cover, then $\opt(y)=0$, so $\mathcal{A}(y)\leq\alpha\cdot 0=0$. If $x$ has no set cover, then $\mathcal{A}(y)\geq \opt(y)\geq 1$. Therefore, $\mathcal{A}$ decides \SC{} in polynomial time.
\end{proof}

\subsection{Proof of~\Cref{t2-hard}}\label{appendix:hardness-generalized}

\nphardtaulessthanone*
\begin{proof}
    Let $\tau\in(\frac{1}{2},1)$, and note that~\Cref{t1-hard} covers the case where $\tau = 1$. Consider an instance of \SC{} with $m$ sets, $n$ elements, and $k\in\mathbb{N}$.
    We first claim that we may assume $\frac{2m+4n+k}{1+2m+4n+k}\geq\tau^2$. Otherwise, we create an equivalent instance of \SC{} with $\lceil\frac{\tau^2}{4n(1 - \tau^2)}\rceil$ copies of each element, giving each copy the same set memberships as the original element.\footnote{Conceptually, one may also think of this procedure as the introduction of $\lceil\frac{\tau^2}{4n(1 - \tau^2)}\rceil$ twins for each element vertex in our constructed \CRI{} instance.}
    Then, letting $n'$ denote the number of elements in our equivalent instance, we have
    \begin{align*}
        n' &\geq \frac{\tau^2}{4(1 - \tau^2)} \\
        4n'(1 - \tau^2) &\geq \tau^2 > \tau^2 + 2m(\tau^2 - 1) + k(\tau^2 - 1) \\
        \frac{2m + 4n' + k}{1 + 2m + 4n' + k} &\geq \tau^2,
    \end{align*}
    so our assumption is safe. We now repeat the same construction given by~\Cref{t1-hard} and depicted in~\Cref{cri-instance}.
    This time, we let the number $X$ of vertices in the independent set adjacent to $a$ be an integer in the interval given by the
    following lemma. We prove that a suitable value always exists.

    \begin{restatable}{claim}{integerintervalclaim}\label{clm:integer}
        For every $m,n,k\in\mathbb{N}$ and $\tau\in\left(\frac{1}{2},1\right)$, there exists a natural number $X$ in the interval $\left(\frac{2+2m+3n-\tau(2+2m-k+2n+1)}{2\tau-1},\frac{2+2m+3n-\tau(2+2m-k+2n)}{2\tau-1}\right]$.
    \end{restatable}
    \begin{proof}
        Taking the difference of the bounds of the interval, we get that the length of the interval is $\frac{\tau}{2\tau-1}$. If $\tau<1$, then $-\tau>-1$ and thus $\tau>2\tau-1$. Thus the length of the interval is greater than 1 when $\tau\in(\frac{1}{2},1)$, so there must exist some integer $X$ within this interval. Furthermore, for $m,n,k\in\mathbb{N}$ and $\tau\in(\frac{1}{2},1)$, the bounds of this interval are  both positive, so $X$ must be a natural number, as desired.
    \end{proof}

    For the forward direction, suppose that there exists a cover $C \subseteq \{S_j\}_{j=1}^m$ of size $k$. Again, let $T = \{bs_j \colon S_j \in C\}$, the set of edges from $b$ to the set vertices corresponding to sets in the cover $C$. Now, $\ccg{G+T}{a}=X+2+2m+3n$ and $\ccg{G+T}{b}=2X+2+2m-k+2n$. We want to show that $\ccrg{G+T}{a}{b}\geq\tau$, but we must still ensure that $\ccrg{G+T}{a}{b}\leq 1$, or else we violate our min/max definition of closeness ratio. Note that this was not a concern for $\tau=1$, as we showed our ratio was exactly $1$ when there was a set cover. Thus we consider the case where $\ccg{G+T}{b}\geq \ccg{G+T}{a}$ and the case where $\ccg{G+T}{b}<\ccg{G+T}{a}$, and show that in both cases $\ccrg{G+T}{a}{b}\geq\tau$.\\\\
    If $\ccg{G+T}{b}\geq \ccg{G+T}{a}$, then $\ccrg{G+T}{a}{b}=\frac{X+2+2m+3n}{2X+2+2m-k+2n}$. Observe that this ratio is greater than $\frac{1}{2}$, and as $X\to\infty$, $\ccrg{G+T}{a}{b}\to\frac{1}{2}$. Thus $\ccrg{G+T}{a}{b}$ is a decreasing function of $X$ and using~\Cref{clm:integer}, we have
    \[\ccrg{G+T}{a}{b}\geq\frac{(\frac{2+2m+3n-\tau(2+2m-k+2n)}{2\tau-1})+2+2m+3n}{2(\frac{2+2m+3n-\tau(2+2m-k+2n)}{2\tau-1})+2+2m-k+2n}\]
    \[=\frac{2+2m+3n-\tau(2+2m-k+2n)+(2\tau-1)(2+2m+3n)}{2(2+2m+3n)-2\tau(2+2m-k+2n)+(2\tau-1)(2+2m-k+2n)}\]
    \[=\frac{2\tau(2+2m+3n)-\tau(2+2m-k+2n)}{2(2+2m+3n)-2+2m-k+2n}=\tau.\]\\\\
    Alternatively, if $\ccg{G+T}{b}<\ccg{G+T}{a}$, then $\ccrg{G+T}{a}{b}=\frac{2X+2+2m-k+2n}{X+2+2m+3n}$. Observe that this ratio is now an \emph{increasing} function of $X$ and using~\Cref{clm:integer} we have,
    \[\ccrg{G+T}{a}{b}>\frac{2(\frac{2+2m+3n-\tau(2+2m-k+2n+1)}{2\tau-1})+2+2m-k+2n}{(\frac{2+2m+3n-\tau(2+2m-k+2n+1)}{2\tau-1})+2+2m+3n}\]
    \[=\frac{2(2+2m+3n)-2\tau(2+2m-k+2n+1)+(2\tau-1)(2+2m-k+2n)}{(2+2m+3n)-\tau(2+2m-k+2n+1)+(2\tau-1)(2+2m+3n)}\]
    \[=\frac{2(2+2m+3n)-2\tau-(2+2m-k+2n)}{2\tau(2+2m+3n)-\tau(2+2m-k+2n+1)}\]
    \[=\frac{4+4m+6n-2\tau-2-2m+k-2n}{4\tau+4m\tau+6n\tau-2\tau-2m\tau+k\tau-2n\tau-\tau}=\frac{(2-2\tau)+2m+4n+k}{\tau+2m\tau+4n\tau+k\tau}.\]
    As $\tau\in(\frac{1}{2},1)$, $2-2\tau>0$, so
    \[\frac{(2-2\tau)+2m+4n+k}{\tau+2m\tau+4n\tau+k\tau}>\frac{2m+4n+k}{\tau+2m\tau+4n\tau+k\tau}=\frac{1}{\tau}(\frac{2m+4n+k}{1+2m+4n+k}).\]
    We have already shown that we may assume $\frac{2m+4n+k}{1+2m+4n+k}\geq\tau^2$. Thus we have that,
    \[\ccrg{G+T}{a}{b}>\frac{1}{\tau}(\frac{2m+4n+k}{1+2m+4n+k})\geq\frac{1}{\tau}(\tau^2)=\tau.\]
    Thus in either case, if there is a set cover of $U$ of size $k$, we can add $k$ edges to $G$ to get a closeness ratio greater than or equal to $\tau$.
    
    For the reverse direction, suppose that there is no collection of $k$ sets which covers $U$. We must prove that there is no set $T \subseteq V^2 \setminus E$ of size at most $k$ such that $\ccrg{G+T}{a}{b}\geq\tau$. Let $T^*$ be a set of (at most) $k$ edge additions such that for all sets $T \subseteq V^2 \setminus E$ of size at most $k$, $\ccg{G+T^*}{b} \leq \ccg{G+T}{b}$. Using the same analysis from the proof of Theorem~\ref{t1-hard}, $\ccg{G+T^*}{b}\geq 2X+2+2m-k+2n+1$, while $\ccg{G+T^*}{a}\leq X+2+2m+3n$. Note that our choice of $X$ (\Cref{clm:integer}) guarantees that $\ccg{G+T^*}{a}\leq \ccg{G+T^*}{b}$, so we now have $\ccrg{G+T^*}{a}{b}\leq\frac{X+2+2m+3n}{2X+2+2m-k+2n+1}$. This ratio is greater than $\frac{1}{2}$, and as $X\to\infty$, $\ccrg{G+T^*}{a}{b}\to\frac{1}{2}$. Thus $\ccrg{G+T^*}{a}{b}$ is a decreasing function of $X$ and we have,
   
    \[\ccrg{G+T^*}{a}{b}<\frac{(\frac{2+2m+3n-\tau(2+2m-k+2n+1)}{2\tau-1})+2+2m+3n}{2(\frac{2+2m+3n-\tau(2+2m-k+2n+1)}{2\tau-1})+2+2m-k+2n+1}\]
    \[=\frac{(2+2m+3n)-\tau(2+2m-k+2n+1)+(2\tau-1)(2+2m+3n)}{2(2+2m+3n)-2\tau(2+2m-k+2n+1)+(2\tau-1)(2+2m-k+2n+1)}\]
    \[=\frac{2\tau(2+2m+3n)-\tau(2+2m-k+2n+1)}{2(2+2m+3n)-(2+2m-k+2n+1)}=\tau.\]
\end{proof}
\subsection{Proof of~\Cref{thm:inapprox}}\label{appendix:hardness-inapprox}
\approxhardness*

\begin{proof}
    Let $\varepsilon > 0$ and $c \geq 1$. Set $\varepsilon' = \frac{\varepsilon}{2(5e^c - 1)}$, and $\delta = \frac{\varepsilon}{2e^c(1 - \varepsilon')}$. Let $(U, S_1, S_2, \ldots, S_m, k)$ be an instance of \SC{}, where $U$ denotes a universe of $n$ elements and $\{S_j\}_{j = 1}^m$ is a set family over $U$.
        We will assume that $n > \frac{ck}{\varepsilon'}$. This assumption is safe, as we may always assume that $k < n$ (otherwise finding a cover is trivial) and so we may always create an equivalent instance which satisfies our assumption by introducing $\lceil\frac{c}{\varepsilon'}\rceil$ copies of each element (where copies have the same set memberships as the original element).
        We construct an instance $(G = (V, E), a, b, k', \tau = 1)$ of \CRI{} as in~\Cref{t1-hard}, with two slight adjustments; see~\Cref{cri-inapprox-instance} for a visual depiction. Specifically, for each element $u_i$, create $m$ twin element vertices and let $X$, the size of the independent set, be $mn+ck$.

        \begin{figure}
            \centering
            \begin{tikzpicture}[ 
                roundnode/.style={circle, draw=black, fill =white, minimum size=8mm},
                goldnode/.style={circle, draw=black, fill=red, minimum size=8mm},
                bluenode/.style={circle, draw=black, fill=red, minimum size=8mm}
            ]
        
            \node[roundnode, fill=blue!20] (a) at (-1.5,3) {$a$};
            \node[roundnode, fill=yellow!60!orange] (b) at (1.5,3) {$b$};
            \node[roundnode] (c) at (0,1) {$c$};
            \draw[-] (a) -- (b);
            \draw[-] (a) -- (c);
            \draw[-] (b) -- (c);
        
            \node[roundnode] (IS1) at (-3.5,4) {$z_1$};
            \node[roundnode] (IS2) at (-3.5,3) {$z_2$};
            \node (dotsIS) at (-3.5,2.35) {$\vdots$};
            \node[roundnode] (ISX) at (-3.5,1.25) {$z_{mn+ck}$};
            \draw[-] (a) -- (IS1);
            \draw[-] (a) -- (IS2);
            \draw[-] (a) -- (ISX);
            \draw[decorate, decoration={brace, amplitude=10pt}, thick]
            ([xshift=-7pt]ISX.south west) -- ([xshift=-12pt]IS1.north west)
            node[midway, xshift=-35pt] {$IS_{mn+k}$};
        
            \node[roundnode] (S1) at (-2,-0.75) {$s_1$};
            \node[roundnode] (S2) at (-.5,-0.75) {$s_2$};
            \node (dots) at (.75,-0.75) {$\cdots$};
            \node[roundnode] (Sm) at (2,-0.75) {$s_m$};
            \draw[-] (c) -- (S1);
            \draw[-] (c) -- (S2);
            \draw[-] (c) -- (Sm);
        
            \node[roundnode] (E11) at (-4,-3) {$v_{1_1}$};
            \node[roundnode] (E12) at (-3,-3) {$v_{1_2}$};
            \node (dotsE) at (-2,-3) {$\cdots$};
            \node[roundnode] (E14m) at (-1,-3) {$v_{1_{m}}$};
        
            \draw[-] (S1) -- (E11);
            \draw[-] (S1) -- (E12);
            \draw[-] (S1) -- (E14m);
        
            \draw[-] (S2) -- (E11);
            \draw[-] (S2) -- (E12);
            \draw[-] (S2) -- (E14m);
        
            \draw[decorate, decoration={brace, amplitude=10pt}, thick]
            ([yshift=-10pt]E14m.south east) -- ([yshift=-10pt]E11.south west) 
            node[midway, yshift=-20pt] {$v_1$};
            \node (dotsE) at (0,-3) {$\cdots$};
            
            \node[roundnode] (En1) at (1,-3) {$v_{n_1}$};
            \node[roundnode] (En2) at (2,-3) {$v_{n_2}$};
            \node (dotsE) at (3,-3) {$\cdots$};
            \node[roundnode] (En4m) at (4,-3) {$v_{n_{m}}$};
        
            \draw[-] (Sm) -- (En1);
            \draw[-] (Sm) -- (En2);
            \draw[-] (Sm) -- (En4m);
        
            \draw[decorate, decoration={brace, amplitude=10pt}, thick]
            ([yshift=-10pt]En4m.south east) -- ([yshift=-10pt]En1.south west) 
            node[midway, yshift=-20pt] {$v_n$};
        
            \end{tikzpicture}
            \caption{A depiction of the construction given by~\Cref{thm:inapprox}. Given an instance of \SC{}, set vertices $s_j$ correspond to sets. Each element vertex $v_i$ is duplicated into $m$ twins. Here, $IS_X$ denotes an independent set of size $mn+ck$, with each vertex adjacent to $a$. We refer to the proof of~\Cref{thm:inapprox} for a formal description of the construction and the accompanying analysis.}
            \label{cri-inapprox-instance}
        \end{figure}
    
        Suppose that there is a set cover $C$ of $U$ with size $k$. Again, let $T = \{bs_j \colon S_j \in C\}$, and now $\ccg{G+T}{b}=2(mn+ck)+1(2)+1(k)+2(m-k)+2(mn)=4mn+2m + (2c-1)k+2$. Similarly, $\ccg{G+T}{a}=(mn+ck)+1(2)+2(m)+3(mn)=4mn+2m+ck+2$. 
        Since, $c \geq 1$, we have $\ccg{G}{a} \leq \ccg{G}{b}$. Then, by~\Cref{obs:adding-to-fractions} and our assumption regarding $n$, we have
        \begin{align*}        
          \ccrg{G+T}{a}{b} = \frac{4mn+2m+ck+2}{4mn+2m + (2c-1)k+2} &> \frac{4mn+2m+ck+2}{4mn+2m + 2ck+2} \\
          &> \frac{n}{n + ck} \\
          &> \frac{n}{n}\cdot \frac{1}{1 + \varepsilon'} > 1 - \varepsilon',
        \end{align*}
        where the final inequality can be verified via a simple proof by contradiction, making use of the fact that $\varepsilon' > 0$. Thus, if a set cover of size $k$ exists, we can always achieve a closeness ratio greater than $1 - \varepsilon'$.
    
        Now suppose that the union of any $ck$ sets covers at most $(1-\frac{1}{e^c} + \delta)n$ elements. The idea of the proof is to give a lower bound on $\ccg{G+T}{b}$, across all sets $T$ of at most $ck$ edge additions. 
        By an analysis identical to~\Cref{claim:star-on-b}, we may assume that all edges in an optimal solution (of size up to $ck$) are incident on $b$, and that all other endpoints are set vertices.
        Thus, given an optimal solution $T$, at most $ck$ set vertices are distance $1$ from $b$, and at most $(1-\frac{1}{e^c} + \delta)mn$ element vertices are distance $2$ from $b$. In particular, there remain $(m - ck)$ set vertices with distance $2$ from $b$, and $(\frac{1}{e^c} - \delta)mn$ element vertices with distance $3$ from $b$. 
        Then we have
        \begin{align*}
            \ccg{G+T}{b} &\geq 2(mn+ck)+2+2m-ck+2(1-\frac{1}{e^c} + \delta)(mn)+3(\frac{1}{e^c} - \delta)(mn) \\
            &=(4+\frac{1}{e^c} - \delta)mn+2m+ck+2.
        \end{align*} 
            
        Now we will analyze $\ccrg{G+T}{a}{b}$. We note that by monotonicity, $\ccg{G+T}{a} \leq \ccg{G}{a}$, so we will use the latter quantity when deriving our bound. We also note that we may assume $ck + 2 < m$, as otherwise all covers of size $ck$ may be checked in polynomial time.
        Substituting, we have
        \begin{align*}
            \ccrg{G+T}{a}{b} &\leq \frac{4mn+2m+ck+2}{(4+\frac{1}{e^c} - \delta)mn+2m+ck+2} \\
            &\leq \frac{4mn+2m+ck+2 + (m - ck - 2)}{(4+\frac{1}{e^c} - \delta)mn+2m+ck+2 + (m - ck - 2)} \quad\quad(\text{using~\Cref{obs:adding-to-fractions}}) \\
            &= \frac{4mn+3m}{(4+\frac{1}{e^c} - \delta)mn+3m} \\
            &= \frac{4n+3}{(4+\frac{1}{e^c} - \delta)n+3} \\
            &\leq \frac{4n+3 + n - 3}{(4+\frac{1}{e^c} - \delta)n+3 + n - 3}    \quad\quad\text{(using~\Cref{obs:adding-to-fractions})} \\
            &\leq \frac{5n}{(5+\frac{1}{e^c} - \delta)n} \\
            &=\frac{5e^c}{5e^c+1 - e^c\delta}.
        \end{align*}
    
        We have shown that if there is a set cover, a solution $T$ of size at most $k$ exists with $\ccrg{G+T}{a}{b} > 1 - \varepsilon'$. Meanwhile, if the union of any $ck$ sets covers at most $(1-\frac{1}{e^c} + \delta)n$ elements, then $\ccrg{G+T}{a}{b}\leq\frac{5e}{5e+1 - e^c \delta}$ for every set $T$ of at most $ck$ edge additions. To complete the 
        proof, it suffices to show that
        \[
          \frac{5e^c}{5e^c + 1 - \varepsilon}\cdot(1 - \varepsilon') \geq \frac{5e^c}{5e^c+1 - e^c\delta},
        \]
        since then a polynomial-time $\frac{5e^c}{5e^c + 1 - \varepsilon}$-approximation (using up to $ck$ edges) could distinguish between the two relevant cases for any~\SC{} instance.
        Simplifying the above, it suffices to show that
        \begin{align*}
            \frac{1 - \varepsilon'}{5e^c + 1 - \varepsilon} &\geq \frac{1}{5e^c+1 - e^c\delta} \\
            (1 - \varepsilon')(5e^c+1 - e^c\delta) &\geq 5e^c + 1 - \varepsilon \\
            \varepsilon &\geq \varepsilon'5e^c - \varepsilon' + e^c\delta - e^c\delta\varepsilon' \\
            \varepsilon &\geq \varepsilon'(5e^c - 1) + \delta e^c(1 - \varepsilon').
        \end{align*}
        The result now follows by substituting our selected values $\delta = \frac{\varepsilon}{2e^c(1 - \varepsilon')}$ and $\varepsilon' = \frac{\varepsilon}{2(5e^c - 1)}$.
    \end{proof}
\subsection{Proof of~\Cref{thm:all-pairs}}\label{appendix:all-pairs}

\allpairshardness*

Some additional notation will make the proof more readable.
Given a graph $G=(V, E)$, a set of vertices $X \subseteq V$, and a vertex $v \in V$, we define $d_{G}(v, X) := \min_{x \in X}d_G(v, x)$ and $\dsum{G}(v, X) := \sum_{x \in X}d_G(v, x)$.

\begin{proof}
Given an instance $(G, k)$ of \domset \space where $G=(V, E)$ is a graph and $k$ is a positive integer, we construct an instance $(\constructGraph , k', \tau)$ of \allpairs. We let $n=|V|$ and assume $n \geq 4$. Let $\tau = \frac{22n-3}{50n-7}$ and $k' = k$. To construct $\constructGraph =(V', E')$, we first copy $G$ and call the copy $G'$. We denote the vertex set of $G'$ as $G'_V$. Then, we add 4 disjoint vertex sets $A, B_1, B_2,$ and $C$. For $i \in \{1, 2\}$, $|A|=|B_i|=n$ and $|C|=13n$. We add edges to $H$ so that $A, B_1, B_2,$ and $C$ are each a clique. Then, we add an arbitrary perfect matching between $G$ and $A$ and between $A$ and $B_i$. Let $X_1, X_2 \subsetneq C$ with $|X_1|=|X_2| = n$ and $X_1 \cap X_2 = \emptyset$. We add edges to form a perfect matching between $X_i$ and $B_i$. Figure \ref{fig:construct} depicts this construction. \par
We claim that $G$ has a dominating set $D$ with $|D| \leq k$ if and only if there exists some $T \subseteq [V']^2 \setminus E'$ with $|T| \leq k$ such that $\frac{\bestcc{\constructGraph +T}{V'}}{\worstcc{\constructGraph +T}{V'}} \geq \tau$. For the forward direction, suppose there exists a dominating set in $G$ with size at most $k$. Then there exists a dominating set $D \subseteq G'_V$ with $|D| \leq k$. Let $X_3 \subsetneq C$ with $X_1 \cap X_2 \cap X_3 = \emptyset$ and $|X_3|=|D|$. Let $T$ be a set of edges which form a perfect matching between $D$ and $X_3$. This is depicted in Figure \ref{fig:construct}. We first make an observation. 

\begin{observation}\label{obs:dsum}
For $v \in V'$ and $K \subsetneq V'$ such that $\constructGraph [K]$ is a clique, let $S(v, K) = \{k \in K: d_{\constructGraph +T}(v, k) = d_{\constructGraph +T}(v, K)\}$. The following holds:
\[\dsum{\constructGraph +T}(v, K) = (d_{\constructGraph +T}(v, K)+1)|K| - |S(v, K)|.\]
\end{observation}

\begin{proof}
    Since $K$ is a clique, the distance from $v$ to any $k\in K$ is either $d_{\constructGraph +T}(v, K)$ or $d_{\constructGraph +T}(v, K) + 1$. Thus, $\dsum{\constructGraph +T}(v, K) < (d_{\constructGraph +T}(v,K)+1)|K| $. To find the value equal to $\dsum{\constructGraph +T}(v, K)$, we need to subtract one from $(d_{\constructGraph +T}(v,K)+1)|K|$ for every $k \in K$ with $d_{\constructGraph +T}(v, k) = d_{\constructGraph +T}(v, K)$.
\end{proof}

With this observation, we calculate an upper bound for the closeness centrality of $g \in G'_V \setminus D$. We assume $g$ is adjacent to exactly one $d \in D$ to get an upper bound for $\dsum{\constructGraph +T}(g, C)$. Note that $\dsum{\constructGraph + T}(g, G'_V) \leq 3(n-1)$, since there is a path of length 3 between any two vertices in $G'_V$ through $A$. With these facts, as well as Observation \ref{obs:dsum}, we can write the upper bound.

\begin{align*}
    \text{cc}_{\constructGraph +T} (g) &= \dsum{\constructGraph +T}(g,G'_V) + \dsum{\constructGraph +T}(g,A) + \dsum{\constructGraph +T}(g,B_1) + \dsum{\constructGraph +T}(g,B_2) + \dsum{\constructGraph +T}(g,C)\\
    & \leq 3(n-1) + (2n-1) + (3n-1) + (3n-1) + (3(13n)-1)\\
    & \leq 50n - 7.
\end{align*}

We argue that for all $v \in V'$, $\text{cc}_{\constructGraph +T}(v) \leq50n - 7$. Suppose $v \in A$. Then, $\dsum{\constructGraph +T}(v, C) = 3(13n) - 2$, $\dsum{\constructGraph +T}(v, B_i) = \dsum{\constructGraph +T}(v, G'_V) = 2n - 1$ and $\dsum{\constructGraph +T}(v, A) = n - 1$. So $\text{cc}_{\constructGraph +T}(v) = 44n - 5 < 50n - 7$. Now, suppose $v \in \{B_1 \cup B_2 \cup C \cup D\}$. Then, $\dsum{\constructGraph +T}(v, C) \leq 2(13n) - 1$. The diameter of $\constructGraph  + T$ is 3. Thus, $\text{cc}_{\constructGraph +T}(v) \leq 26n - 1 + 3(4n) = 38n - 1 < 50n - 7$. So $\worstcc{\constructGraph +T}{V'} \leq 50n - 7$. 

Next, we calculate a lower bound for the closeness centrality of a vertex in $C$. Let $c_s \in C$ be a vertex with the smallest closeness centrality in $C$. Clearly, $c_s$ is adjacent to a vertex in $B_i$ or $D$. If $c_s$ is adjacent to a vertex in $B_1$, then $\dsum{\constructGraph +T}(c_s,B_1) = 2n - 1$, $\dsum{\constructGraph +T}(c_s,B_2) = 2n$, and $\dsum{\constructGraph +T}(c_s,G'_V) = 3n-k$ since there is a path of length 2 to every vertex in $D$ (a dominating set for $G'$). However, if $c_s$ is adjacent to $D$, $\dsum{\constructGraph +T}(c_s,B_i) = 2n$ and $\dsum{\constructGraph +T}(c_s,G'_V) \geq 3n - k -1$. So $\dsum{\constructGraph +T}(c_s,\{B_1 \cup B_2 \cup G'_V\}) \geq 6n - 1$.
\begin{align*}
    \text{cc}_{\constructGraph +T} (c_s) &= \dsum{\constructGraph +T}(c_s,\{B_1 \cup B_2 \cup G'_V\}) + \dsum{\constructGraph +T}(c_s,A) + \dsum{\constructGraph +T}(c_s,C)\\
    & \geq (6n - 1) + (3n - 1) + (13n - 1)\\
    &  \geq 22n - 3.
\end{align*}
Thus, for any $c \in C$, $\text{cc}_{\constructGraph +T} (c) \geq \text{cc}_{\constructGraph +T} (c_s) \geq 22n - 3$. We argue that for all $v \in V'$, $\text{cc}_{\constructGraph +T}(v) \geq 22n - 3$. Note $v \in V' \setminus C$ is adjacent to at most one $c \in C$. Thus, $\dsum{\constructGraph +T}(v,C) \geq 2(13n) - 1 > 22n - 3$. So $\bestcc{\constructGraph +T}{V'} \geq 22n - 3$. 

So we have $22n - 3 \leq \bestcc{\constructGraph +T}{V'}$ and $50n - 7 \geq \worstcc{\constructGraph +T}{V'}$. We now see that $\tau$ is a lower bound on the closeness centrality ratio.
\begin{align*}
    & \tau = \frac{22n - 3}{50n - 7} \leq \frac{\bestcc{\constructGraph +T}{V'}}{50n - 7} \leq \frac{\bestcc{\constructGraph +T}{V'}}{\worstcc{\constructGraph +T}{V'}}.
\end{align*}

\par

\begin{figure}
\centering

\begin{tikzpicture}[scale=1.2]

    \begin{scope}
        \coordinate (C) at (0,3);
        \coordinate (B1) at (-1.5,1.5);
        \coordinate (B2) at (1.5,1.5);
        \coordinate (A) at (0,0);
        \coordinate (G) at (0,-1.5);

        \newcommand{\fourcurves}[2]{
            \draw (#1) to[bend left=25] (#2);
            \draw (#1) to[bend left=10] (#2);
            \draw (#1) to[bend right=10] (#2);
            \draw (#1) to[bend right=25] (#2);
        }

        \fourcurves{A}{G}
        \fourcurves{A}{B1}
        \fourcurves{A}{B2}
        \fourcurves{B1}{C}
        \fourcurves{B2}{C}

        \node[circle, draw, minimum size=2cm, fill=white] at (C) {$C$};
        \node[circle, draw, minimum size=1.25cm, fill=white] at (B1) {$B_1$};
        \node[circle, draw, minimum size=1.25cm, fill=white] at (B2) {$B_2$};
        \node[circle, draw, minimum size=1.25cm, fill=white] at (A) {$A$};
        \node[circle, draw, minimum size=1.25cm, fill=white] at (G) {$G$};
    \end{scope}

    \begin{scope}[xshift=6cm]
        \newcommand{\fourcurves}[2]{
            \draw (#1) to[bend left=25] (#2);
            \draw (#1) to[bend left=10] (#2);
            \draw (#1) to[bend right=10] (#2);
            \draw (#1) to[bend right=25] (#2);
        }
    
        \coordinate (C2) at (0,3);
        \coordinate (B12) at (-1.5,1.5);
        \coordinate (B22) at (1.5,1.5);
        \coordinate (A2) at (0,0);
        \coordinate (G2) at (0,-1.5);
        \coordinate (D) at (G2)+(.2, 0);

        \fourcurves{A2}{G2}
        \fourcurves{A2}{B12}
        \fourcurves{A2}{B22}
        \fourcurves{B12}{C2}
        \fourcurves{B22}{C2}

        \node[circle, draw, minimum size=1.25cm, fill=white] at (B12) {$B_1$};
        \node[circle, draw, minimum size=1.25cm, fill=white] at (B22) {$B_2$};
        \node[circle, draw, minimum size=1.25cm, fill=white] at (A2) {$A$};
        \node[circle, draw, minimum size=1.25cm, fill=white] at (G2) {};

        \foreach \dx in {0, .15, 0.3} {
            \draw[red]
            ($(G2)+(\dx / 2,-\dx / 2)$)
            .. controls ($(3.25+\dx,-.5-\dx)$) and ($(3+\dx,3+\dx)$)
            .. ($(C2)+(\dx,\dx)$);
        }

        \node[circle, draw, minimum size=2cm, fill=white] at (C2) {$C$};

        \node[circle, draw,                   minimum size=0.7cm, fill=white] 
        at ($(G2)+(0.2,0)$) {$D$};

        \node[right=2.25cm of C2, red] {$T$};

    \end{scope}

\end{tikzpicture}

\caption{On the left is the unmodified construction of $\constructGraph $ given by Theorem \ref{thm:all-pairs}. On the right is $\constructGraph $ with the addition of an edge set $T$. Given a dominating set $D$ in $G$, this choice of $T$ allows us to achieve a ratio of at least $\tau$ between the vertex with the minimum closeness centrality and the vertex with the maximum closeness centrality in $\constructGraph +T$.}
\label{fig:construct}

\end{figure}

We now prove the backward direction. Assume that there does not exist a dominating set in $G$ (and hence, in $G'$) with size at most $k$. We will show that for any edge set $T$, $\frac{\bestcc{\constructGraph +T}{V'}}{\worstcc{\constructGraph +T}{V'}} < \tau$. We begin by proving the following claim.

\begin{claim} \label{clm:structure}
    For any set $T^* \subseteq [V']^2 \setminus E$ with size at most $k$, there exists a set $T \subseteq [V']^2 \setminus E$ of size at most $k$ such that every edge in $T$ is in either $\edgetype{G'_V}{C}$, $\edgetype{G'_V}{A}$, or $\edgetype{A}{C}$ and $d_{\constructGraph +T}(g, c) \leq d_{\constructGraph  + T^*}(g, c)$ for all $g \in G'_V$ and $c \in C$.
\end{claim}

\begin{proof}
Fix some $T^*$. First, we set $T' \gets T^*$. We will modify $T'$ to exclude all edges in 
$\edgetype{G'_V}{B_i}$, $\edgetype{A}{B_i}$, $\edgetype{B_1}{B_2}$, $\edgetype{G'_V}{G'_V}$, and $\edgetype{B_i}{C}$ while maintaining $|T'| \leq k$. Then, we will set $T \gets T'$ and argue that $d_{\constructGraph +T} (g, c) \leq d_{\constructGraph +T^*} (g, c)$ for all $g \in G'_V$ and $c \in C$. \par

First, for each edge $gb \in \edgetype{G'_V}{B_i}$ in $T'$, we replace $gb$ with $gc$, where $c \in C$ is arbitrary. Likewise, for each edge $ab \in \edgetype{A}{B_i}$ in $T'$, we replace $ab$ with $ac$, where $c \in C$ is arbitrary. Next, we remove all edges in $\edgetype{B_1}{B_2}$ from $T'$. Then, for each edge $g_1g_2 \in \edgetype{G'_V}{G'_V}$ in $T'$, if $g_1$ is not adjacent to some $c \in C$, we replace $g_1g_2$ with $g_1c$ in $T'$. Otherwise, we replace $g_1g_2$ with $g_2c$. Now we consider edges in $T'$ which are in $\edgetype{B_i}{C}$. Note that for every $b\in B_i$, there is exactly one vertex $g_b \in G'_V$ in the unmodified graph $\constructGraph$ with $d_{\constructGraph }(g, b)=2$. For every edge $bc \in \edgetype{B_i}{C}$ in $T'$, we replace $bc$ with $g_bc$. Finally, we set $T \gets T'$. Note $T$ contains only edges in $\edgetype{G'_V}{C}$, $\edgetype{G'_V}{A}$, and $\edgetype{A}{C}$.\par

Consider some vertex pair $(g, c)$ with $g \in G'_V$ and $c \in C$. We will verify that $d_{\constructGraph +T}(g, c) \leq d_{\constructGraph +T^*}(g, c)$. If there is a shortest path in $\constructGraph +T^*$ which does not include any edge in $T^*$, this path still exists in $\constructGraph +T$, so the inequality is fulfilled. So assume that all shortest paths between $g$ and $c$ in $\constructGraph +T^*$ contain some $uv \in T^*$. Additionally, since $d_{\constructGraph }(g, c) \leq 4$, these paths must have length at most 3. Consider one such path $P$ from $g$ to $c$ in $\constructGraph +T^*$ with at least edge $uv \in T^*$. We will show there is some path from $g$ to $c$ in $\constructGraph +T$ with length at most the length of $P$. \par

Suppose $uv \in \edgetype{G'_V}{B_i}$. If $g$ is an endpoint of $uv$, then there is a path $gc'c$ in $\constructGraph +T$ for some $c' \in C$, which must be at most the length of the shortest path in $\constructGraph +T^*$. If $g$ is not an endpoint of $uv$, then we have a path $gg'bc$ in $\constructGraph +T^*$ for some $g'\in G'_V$ and $b \in B_1 \cup B_2$. Since we have the edge $g'b$ in $\constructGraph +T^*$, we have the edge $g'c'$ for some $c' \in C$ in $\constructGraph +T$. Since this replacement was made first, the edge $gg'$ was replaced with $gc''$ for some $c'' \in C$. Therefore, we have a path $gc''c$ in $\constructGraph +T$.\par

Suppose $uv \in \edgetype{A}{B_i}$. This path must be $gabc$ for some $a \in A$ and $b \in B_1 \cup B_2$. Since $ga$ is still present in $\constructGraph +T$, there is a path $gac'c$ of the same length for some $c' \in C$ in $\constructGraph +T$.\par

Suppose $uv \in \edgetype{B_1}{B_2}$. This path must be $gb_1b_2c$, for some $b_1 \in B_1$ and $b_2 \in B_2$. Since this path contains the edge $gb_1 \in \edgetype{G'_V}{B_1}$, we have a shorter path $gc'c$ for some $c' \in C$ in $\constructGraph +T$. \par 

Suppose $uv \in \edgetype{G'_V}{G'_V}$. If the path is $gg'ac$ for $g' \in G'_V$ and $a \in A$, there is another path $ga'ac$ for $a, a' \in A$. If the path is $gg'c'c$ for $g' \in G'_V$ and $c' \in C$, then there is a path $gc''c$ for some $c'' \in C$ in $\constructGraph +T$. If the path was $gg'g''c$ for $g', g'' \in G'_V$, there is a path $gg'c$ in $\constructGraph +T$. If the path is $gg'bc$ for some $g' \in G'_V$ and $b \in B_1 \cup B_2$, then this case was previously covered. \par

Suppose $uv \in \edgetype{B_i}{C}$. The cases of the paths $gbc$ and $gg'bc$ for some $g' \in G'_V$ and $b \in B_1 \cup B_2$ were already covered. So we will consider a path $gabc$ for some $a \in A$ and $b \in B_1 \cup B_2$. If in the unmodified graph $\constructGraph $ we have $d_{\constructGraph }(g, b)=2$, then the path $gc'c$ exists for some $c' \in C$ in $\constructGraph +T$. Additionally, if $ab \notin \constructGraph $, then this case was previously covered. Otherwise, we must have $ga \notin \constructGraph $. Then, there must be some vertex $g'  \in G'_V$ with $d_{\constructGraph }(g', a)=1$. Therefore, $d_{\constructGraph }(g', b)=2$. So the path $gag'c$ is present in $\constructGraph +T$.\par
We have verified that every shortest path $P$ in $\constructGraph +T^*$ between two vertices $g \in G'_V$ and $c \in C$ which is not present in $\constructGraph $ has a corresponding path in $\constructGraph +T$ which has length at most the length of $P$. So the inequality holds and the proof is complete.
\end{proof}
Claim \ref{clm:structure} allows us to restrict the structure of the edge set $T$ without increasing $d_{\constructGraph +T}(g, c)$ for any $g \in G'_V$ and $c \in C$ pair. We will use this structure in Claim \ref{clm:farC} to get a lower bound on $\dsum{\constructGraph +T}(v, C)$ for some $v \in G'_V$, allowing us to obtain a sufficiently large lower bound on $\worstcc{\constructGraph +T}{V'}$.
\begin{claim}\label{clm:farC}
    If there is no dominating set of size $k$ in $G'_V$, then for every set of edge additions $T \subseteq V^2 \setminus E$ in $\constructGraph $, there exists some $v \in G'_V$ such that $\dsum{\constructGraph +T}(v, C) \geq 52n - k -2$. 
\end{claim}

\begin{proof}
Let $T \subseteq [V']^2 \setminus E'$ be an arbitrary set of edge additions of size at most $k$. By \cref{clm:structure}, we may assume that $T$ contains only edges in $\edgetype{G'_V}{C}$, $\edgetype{A}{C}$, or $\edgetype{G'_V}{A}$ without increasing $d_{\constructGraph +T}(g, c)$ for any $g \in G'_V$ and $c \in C$. Let $S_G = \{g\in G'_V : gv \in T\}$ and $S_A = \{a \in A : av \in T\}$. For each $g \in G'_V$, denote $a_g \in A$ as the single vertex in $A$ adjacent to $g$ in the unmodified construction $\constructGraph $. Then, let $S_{A'} = \{g \in G'_V : a_g \in S_A\}$. Finally, let $S = S_G \cup S_{A'}$.

Since $T$ contains no edge with both endpoints in $G'_V$ or both endpoints in $A$, $|S| \leq k$. Since $G'_V$ has no dominating set with size at most $k$, there exists a vertex $v \in G'_V$ such that $d_{\constructGraph +T}(v, S) \geq 2$. Note in the unmodified $\constructGraph $, $v$ has distance 3 from two vertices in $C$ and distance 4 from $13n-2$ vertices in $C$. \par

Consider the set $W = \{c \in C : d_{\constructGraph  + T}(v, c) < d_{\constructGraph }(v, c)\}$. The shortest path $P$ from $v$ to $c \in W$ must contain an edge in $T$. Additionally, the length of $P$ is at most 3. If $P$ contains some $s \in S$, then $P = vgsc$ for some $g \in G'_V$, so the length of $P$ is 3.\par

If $P$ does not contain some $s \in S$, then $P$ contains some edge in $\edgetype{A}{C}$. However, we cannot have a path of length 2 since this would imply $v \in S$. Therefore, the length of $P$ is 3.\par

Since every $c \in W$ is an endpoint of an edge in $T$, we have $|W| \leq k$. Since $d(v, c)=3$, we can write the following:

\begin{center}
        $\dsum{\constructGraph +T}(v, C) \geq 3(k+2) + 4(13n-(k+2)) = 52n - k -2.$
\end{center}

\end{proof}

To finish the proof, we use Claim \ref{clm:farC} to bound $\worstcc{\constructGraph +T}{V'}$ from below. Then, we bound $\bestcc{\constructGraph +T}{V'}$ from above and use these bounds to show that the ratio between $\bestcc{\constructGraph +T}{V'}$ and $\worstcc{\constructGraph +T}{V'}$ is less than $\tau$. \par

Fix an arbitrary edge set $T \in [V']^2\setminus E'$.
Let $v \in G'_V$ be a vertex which obtains $\dsum{\constructGraph +T}(v, C) \geq 52n - k-2$. Note that if $k$ vertices in $A$ are adjacent to $v$, we get $\dsum{\constructGraph +T}(v, A) = 2n-k-1$. Additionally, if $k$ vertices in $B_i$ are adjacent to $v$, we get $\dsum{\constructGraph +T}(v, B_i) = 2n-k$. We can now bound the closeness centrality of $v$ from below.

\begin{align*}
cc_{\constructGraph  + T}(v)= &\dsum{\constructGraph +T}(v, G'_V) +  \dsum{\constructGraph +T}(v, A) +  \dsum{\constructGraph +T}(v, B_1) +  \dsum{\constructGraph +T}(v, B_2) +  \dsum{\constructGraph +T}(v, C) \\
& \geq (n-1) + (2n-k-1) + 2(2n - k) + (52n - k -2) \\
& = 59n - 4k - 4. \\
\end{align*}
Therefore, we must have $\worstcc{\constructGraph +T}{V'} \geq 59n-4k-4$.\par
Next, we get an upper bound on the closeness centrality of a vertex $c \in C$ in $\constructGraph +T$. We bound the closeness centrality of $c$ in $\constructGraph $ since the closeness centrality could only be smaller in $\constructGraph +T$. We assume $c$ is adjacent to some $b \in B_i$.
\begin{align*}
\dsum{\constructGraph +T}(c, \constructGraph ) = &\dsum{\constructGraph +T}(c, G'_V) +  \dsum{\constructGraph +T}(c, A) +  \dsum{\constructGraph +T}(c, B_1) +  \dsum{\constructGraph +T}(c, B_2) +  \dsum{\constructGraph +T}(c, C) \\
& \leq (4n) + (3n-1) + (2n-1) + (2n) +  (13n-1) \\
& =24n - 3. \\
\end{align*}
Therefore, we must have $\bestcc{\constructGraph +T}{V'} \leq 24n-3$. From these bounds, we complete the proof.
 \[\frac{\bestcc{G'_V+T}{V'}}{\worstcc{G'_V+T}{V'}} \leq \frac{24n - 3}{59n - 4k - 4} \leq \frac{24n-3}{55n-4} < \tau = \frac{22n-3}{50n-7}\]
 since $n \geq 4$.
\end{proof}

\bibliographystyle{abbrvnat}
\bibliography{refs}

\end{document}